\documentclass[sigconf, nonacm]{acmart}

\usepackage{algorithmic}
\usepackage[ruled,vlined]{algorithm2e}
\usepackage{booktabs}
\usepackage{pifont}
\usepackage{fancyhdr} 
\usepackage{color}
\usepackage{soul, framed}
\usepackage{makecell}
\usepackage{xcolor}
\usepackage{colortbl}
\usepackage{graphicx}
\usepackage{textcomp}
\usepackage{xcolor,colortbl}
\usepackage{multirow}
\usepackage{bm,dsfont}
\usepackage{subfigure}
\usepackage{balance}
\usepackage{stmaryrd}
\usepackage{autobreak}
\usepackage{enumitem}
\newtheorem{definition}{Definition}
\newtheorem{lemma}{Lemma}

\newtheorem{example}{Example}

\newcommand{\netname}{\textit{ALICE}}
\newcommand{\ksize}{\textit{k}}
\newcommand{\aggname}{\textit{ConNet}}

\AtBeginDocument{%
  \providecommand\BibTeX{{%
    \normalfont B\kern-0.5em{\scshape i\kern-0.25em b}\kern-0.8em\TeX}}}

\setcopyright{acmcopyright}
\copyrightyear{2024}
\acmYear{2024}
\acmDOI{XXXXXXX.XXXXXXX}

\acmConference[SIGMOD '24]{Make sure to enter the correct
  conference title from your rights confirmation emai}{June 9--15,
  2024}{Santiago, Chile}

\acmBooktitle{Proceedings of the 2024 International Conference on Management of Data (SIGMOD ’24), June 9--15, 2024, Santiago, Chile} 
\acmPrice{15.00}
\acmISBN{978-1-4503-XXXX-X/18/06}

\begin{document}

\title{Neural Attributed Community Search at Billion Scale
}

\author{Jianwei Wang$^{1}$ ~~ Kai Wang$^{2}$ ~~ Xuemin Lin$^{2}$~~Wenjie Zhang$^{1}$ ~~Ying Zhang$^{3}$ ~~\\
 \normalsize{$^{1}${The University of New South Wales, Sydney, Australia}}\\
 \normalsize{$^{2}${Shanghai Jiao Tong University, Shanghai, China}} \\
 \normalsize{$^{3}${University of Technology Sydney, Sydney, Australia}} \\
\normalsize{jianwei.wang1@unsw.edu.au, w.kai@sjtu.edu.cn, xuemin.lin@sjtu.edu.cn, zhangw@cse.unsw.edu.au, ying.zhang@uts.edu.au}\\
}

\renewcommand{\shortauthors}{Jianwei Wang et al.}
\renewcommand{\shorttitle}{Research Data Management Track Paper}

\begin{abstract}
  Community search has been extensively studied in the past decades.
In recent years, there is a growing interest in attributed community search that aims to identify a community based on both the query nodes and query attributes. A set of techniques have been investigated. Though the recent methods based on advanced learning models such as graph neural networks (GNNs) can achieve state-of-the-art performance in terms of accuracy, we notice that 1) they suffer from severe efficiency issues; 2) they directly model community search as a node classification problem and thus cannot make good use of interdependence among different entities in the graph.
Motivated by these, in this paper, we propose a new neur\underline{\textbf{AL}} attr\underline{\textbf{I}}buted \underline{\textbf{C}}ommunity s\underline{\textbf{E}}arch model for large-scale graphs, termed \netname. 
\netname~first extracts a candidate subgraph to reduce the search scope and subsequently predicts the community by the \underline{Con}sistency-aware \underline{Net}, termed \aggname. 
Specifically, in the extraction phase, we introduce the density sketch modularity that uses a unified form to combine the strengths of two existing powerful modularities, i.e., classical modularity and density modularity. Based on the new modularity metric, we first adaptively obtain the candidate subgraph, formed by the $k$-hop neighbors of the query nodes, with the maximum modularity. Then, we construct a node-attribute bipartite graph to take attributes into consideration.
After that, \aggname~adopts a cross-attention encoder to encode the interaction between the query and the graph. The training of the model is guided by the structure-attribute consistency and the local consistency to achieve better performance.
Extensive experiments over 11 real-world datasets including one billion-scale graph demonstrate the superiority of \netname~ in terms of accuracy, efficiency, and scalability. Notably, \netname~ can improve the F1-score by 10.18\% on average and is more efficient on large datasets in comparison to the state-of-the-art. \netname~ can finish training on the billion-scale graph within a reasonable time whereas state-of-the-art can not.

\end{abstract}

\keywords{Attributed Community Search; Graph Neural Networks} 

 \maketitle

\vspace{-3mm}
\section{Introduction}
\label{sec:intro}

\begin{figure}
  \centering
  \includegraphics[width=0.80\linewidth]{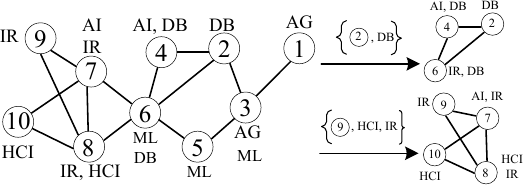}
  \vspace{-4mm}
  \caption{An illustration of attributed community search: The left panel illustrates a citation graph, whereas the right panel displays the retrieved communities corresponding to queries on each arrow.}
  \label{fig:case}
\vspace{-5mm}
\end{figure}

Graph-structured data has shown particular advantages for modeling relationships and dependencies between objects, making it a powerful tool for data analytics in various fields such as social networks~\cite{fan2019graph, newman2002random, wang2023towards, zhang2018finding}, biological networks~\cite{zhang2021graph, hetzel2021graph} and financial networks~\cite{xiang2022temporal, cheng2020knowledge}. One of the core tasks of graph analytics is community search (CS)\cite{cui2013online, fang2020survey, hu2016querying, huang2019community, zhang2018finding} that aims to find a subgraph containing the specific query nodes, with the resulting subgraph (community) being a densely intra-connected structure. 
{\color{black} In many real-world applications, nodes are often associated with attributes~\cite{pfeiffer2014attributed, wang2021survey}.} As such, it is desirable to query using not just query nodes, but also query attributes. 
Attributed Community Search (ACS)~\cite{fang2016effective, fang2020survey}, a related but more challenging problem compared to CS, is proposed to deal with such applications. ACS aims to identify a community based on both query nodes and attributes, with the resulting community expected to demonstrate {\color{black}structure cohesiveness and semantic homogeneity}. {\color{black} Studying ACS can benefit various applications, e.g., extracting biologically significant clues of protein-protein interaction networks~\cite{huang2017attribute, bhowmick2015clustering}, finding the research communities in the collaboration networks~\cite{fang2016effective}, detecting fraudulent keywords of the web search~\cite{yang2021exploiting}.}
In light of the significance and popularity of ACS, a spectrum of algorithms~\cite{fang2016effective, huang2017attribute, jiang2022query, gao2021ics} have been developed, which can be classified into two categories: non-learning-based techniques and learning-based approaches.

\begin{figure*}
  \centering
  \includegraphics[width=0.96\linewidth]{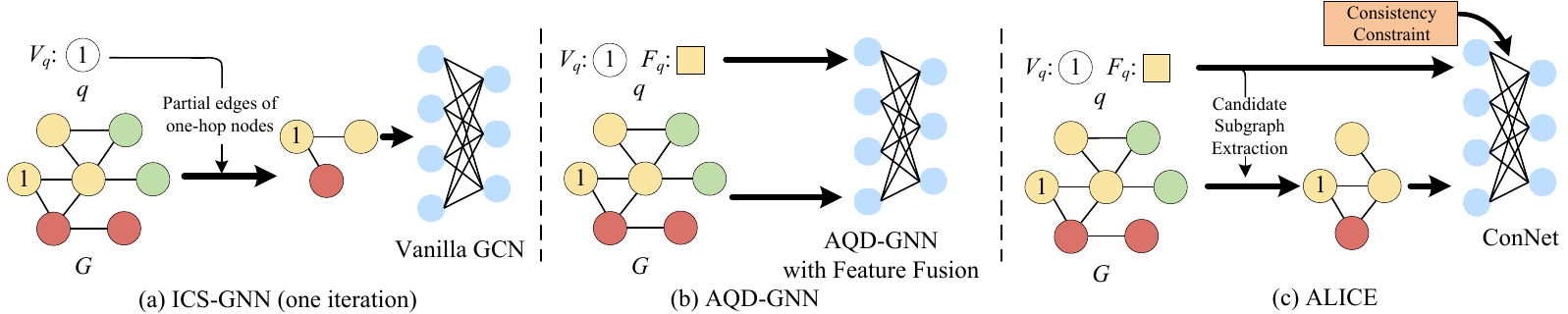}
  \vspace{-3mm}
  \caption{The framework of learning-based (attributed) community search models}
  \label{fig:sota_compare}
  \vspace{-4mm}
\end{figure*}

\vspace{1mm}
\noindent \textbf{\color{black}{Existing solutions.}} 
Existing non-learning-based attributed community search algorithms~\cite{fang2016effective, huang2017attribute} use a decoupled scheme that treat structure and attribute separately. They first search for structural-cohesive nodes based on the pre-defined cohesive subgraph models such as \ksize-core~\cite{fang2016effective} and \ksize-truss~\cite{huang2017attribute}. Subsequently, the algorithms compute the score of attribute cohesiveness to identify the most relevant communities.
These non-learning-based methods, however, are constrained by two primary limitations~\cite{jiang2022query}: 1) Structure inflexibility. The pre-defined subgraph models rely heavily on the hyper-parameter \ksize~and the community quality is sensitive to \ksize. In addition, the fixed subgraph models place a highly rigid constraint on the topological structure of communities, making it difficult for real-world communities to meet such priors. 2) Attribute irrelevance. These algorithms consider each attribute independently, which fails to capture the latent correlations between attributes, and thus restricts the exploration capability in the semantic space.

In order to alleviate the above issues, learning-based techniques with Graph Neural Networks (GNNs) have been proposed including ICS-GNN~\cite{gao2021ics} and AQD-GNN~\cite{jiang2022query}. Their frameworks are illustrated in Figure~\ref{fig:sota_compare}(a) and Figure~\ref{fig:sota_compare}(b), respectively. They refrain from imposing constraints on the community structure, and the attributes are propagated through edges to enhance their connection. ICS-GNN is designed for interactive community search that aims to gradually find the community in multiple iterations. In each iteration, a Vanilla Graph Convolutional Network (GCN) model~\cite{DBLP:conf/iclr/KipfW17} is directly applied to the one-hop neighbors of the query nodes. AQD-GNN is proposed to support ACS that inputs both the query \textit{q} and data graph \textit{G}, and a feature fusion operator is utilized to combine the representations of the query node, query attribute, and data graph to predict the final community.

Although the above-mentioned learning-based approaches demonstrate remarkable performance, particularly AQD-GNN, which has achieved state-of-the-art accuracy for ACS as evaluated in~\cite{jiang2022query}, two main limitations persist for existing learning-based approaches. 

Firstly, both ICS-GNN and AQD-GNN encounter significant efficiency problems. ICS-GNN requires the entire model to be retrained when a new query is received. However, training a model is a time-consuming process. In addition, AQD-GNN takes the entire graph as input to learn the representation of each node and searches the entire graph to determine whether nodes belong to the community or not. Both learning and searching the complete graph can be time-consuming, hindering efficiency and scalability.

Secondly, both ICS-GNN and AQD-GNN directly recast the community search as a node classification problem, while the interdependence among different entities remains insufficiently explored: 1) The intricate interaction between query and data graph is overlooked, despite their strong correlation w.r.t the final community. ICS-GNN only inputs the candidate subgraph for GNNs, missing the query. {{\color{black} AQD-GNN encodes query nodes, query attributes, and the graph separately, which employs a fusion operator to concatenate them. Nevertheless, the interaction between the query and each node in the graph remains insufficiently explored.}
2) The correlation between structure and attribute remains inadequately investigated. Although they encode both graph structure and attribute simultaneously, they still fall short in capturing the correlation between the representations of query nodes and attributes, which is essential for structure and attribute constraints. 
3) Both methods typically disregard the connection among nodes within a community. The community is a cohesive component that considers multiple nodes together, whereas node classification focuses on individual node properties. While it is possible to use breath first search (BFS) to select nearby connected nodes after learning, it can still harm accuracy since the optimization signal cannot be backpropagated.

Therefore, to design an efficient and effective learning-based approach for ACS at a large scale, two main challenges exist below.

\textit{Challenge 1: How to efficiently perform learning-based ACS at a large scale?} A direct way is discarding unpromising nodes/edges at an early stage. However, the size of real-world communities can differ a lot. If we limit our selection to a small portion of nodes as candidates (e.g., ICS-GNN solely depends on one-hop neighbors of the query node), we risk losing many promising nodes; otherwise, an overly broad search scope presents computational difficulties. Thus, how to adaptively select promising candidates while taking both structure and attribute into account is challenging.

\textit{Challenge 2: How to effectively exploit the interdependence among different entities to enhance prediction accuracy?} 
There exist abundant entities for ACS including query, data graph, structure, and attribute, while some of them (like structure and attribute) are from heterogeneous spaces ~\cite{chen2020learning}. Thus, how to collaboratively utilize these entities, capture their interactions and consistency in the latent space, and improve the overall accuracy is challenging.

\vspace{1mm}
\noindent \textbf{Our solutions.}
To tackle the above challenges, in this paper, we propose a new neur\underline{\textbf{AL}} attr\underline{\textbf{I}}buted \underline{\textbf{C}}ommunity s\underline{\textbf{E}}arch model for large-scale graphs, namely \netname~(in Figure~\ref{fig:sota_compare}(c)). {\color{black}} \netname~ is a two-stage framework that first extracts a candidate subgraph and subsequently searches the community over the candidate subgraph by the \underline{Con}sistency-aware \underline{Net} (\aggname).

\vspace{1mm}
\noindent \textit{(1) Candidate Subgraph Extraction}. 
To address Challenge 1, it is crucial to find an effective model for evaluating the cohesiveness of a subgraph. Here we resort to modularity, a parameter-free metric that has been extensively utilized in finding communities~\cite{fortunato2007resolution, barber2007modularity, guo2023resolution, DBLP:conf/sigmod/KimLCY22}. 
{\color{black}However, existing proposed modularities either select too many loosely connected nodes due to the free-rider effect~\cite{wu2015robust} and the resolution limit problem~\cite{fortunato2007resolution}, like the classical modularity~\cite{newman2004finding}; or impose overly stringent requirements on cohesiveness, which may hinder the exploration of promising nodes, like the density modularity~\cite{DBLP:conf/sigmod/KimLCY22}. 
To alleviate this situation, we propose a novel form of modularity called \textit{density sketch modularity} that uses a unified form to balance and combine the strengths of the above two modularities. }
{\color{black}Then, we adaptively select the structure-based candidate subgraph $H$ induced by the $k$-hop neighbors of the query nodes s.t. $H$ has the maximum density sketch modularity value. In this way, we do not need to pre-set the value of $k$.
Similarly, to select nodes that possess similar attributes to the query attributes, we construct a node-attribute bipartite graph. With the bipartite graph, we select the attribute-based candidate nodes based on the subgraph that is induced by the $k$-hop neighbors of the query attributes and has the largest bipartite modularity value.} 

\vspace{1mm}
\noindent \textit{(2) Consistency-aware Net}. To address Challenge 2, we further devise a novel GNN-based consistency-aware net, namely \aggname~{\color{black}to capture the correlation and consistency}. It has three main components: \\
1) \textit{Cross-attention encoder}. We design a cross-attention encoder that aims to weigh the correlation between each query node (\textit{resp.} attribute) and graph node (\textit{resp.} attribute) and utilize this correlation to learn the structure (\textit{resp.} attribute) representation. {\color{black}In contrast to AQD-GNN, which encodes query locally in one layer, \aggname~learns a representation that effectively combines the interaction of the query and each node in the graph.}
2) \textit{Structure-attribute consistency}. We devise a structure-attribute consistency module inspired by the recent representation learning methodology that brings related entities closer together in the latent space \cite{mikolov2013distributed, radford2021learning, DBLP:conf/iclr/ShiaoG0PLS23}. {\color{black}In light of the high correlation between structure and attribute}, we propose a new approach that aims to minimize the Wasserstein distance between the distribution of structure representation and the distribution of attribute representation. 
3) \textit{Local consistency}. We develop a local consistency module, based on the observation that if a node belongs to a community, its neighboring nodes exhibit a high likelihood of being part of the same community and vice versa.  It aims to pull closer the prediction results of nodes that are linked together. ACS is then modeled as multi-task learning that signals from the ground-truth labels and signals from the two consistency constraints are then optimized together.

\vspace{1mm}
\noindent \textbf{Contributions.}
Here we summarize our main contributions:
\begin{itemize}[leftmargin=10pt, topsep=1pt]
\item{} To enhance the performance of ACS, we propose a novel learning-based method \netname~ that first extracts promising candidate subgraph and subsequently searches communities by the \aggname.
\item{} We design an efficient subgraph extraction algorithm by leveraging a new form of modularity (i.e., density sketch modularity) and node-attribute relationship to adaptively select promising nodes. {\color{black}As evaluated, our approach can significantly reduce the training graph size (e.g., on the \textit{Orkut} dataset with {\color{black}3.07M} nodes, only about {\color{black}$<$1\%} of nodes need to be passed to the next stage).} 
\item{} We propose a GNN-based model \aggname~to preserve both structure-attribute consistency and local consistency among nodes. It employs a cross-attention encoder to effectively capture the interaction between the query and the data graph.
\item{} Extensive experiments are conducted over 11 popular public datasets, encompassing one billion-scale graph \textit{Friendster}. The results demonstrate that \netname~ can substantially improve both the search accuracy and the efficiency compared with existing methods. {\color{black}It can elevate the F1-score by 10.18\% on average under the setting of query attributes generated from query nodes and is more efficient on large datasets \textit{Google+} and \textit{PubMed} compared with AQD-GNN~\cite{jiang2022query}}. Moreover, \netname~ can finish training on large datasets \textit{Reddit}, \textit{Orkut} and \textit{Friendster} within a reasonable time, whereas ADQ-GNN can not. 

\end{itemize}

\vspace{1mm}
\noindent \textbf{Roadmap.}
Section~\ref{sec:pre} introduces the preliminaries. Section~\ref{sec:overview} gives an overview of the whole framework while Section~\ref{sec:Candidate_Subgraph_Extraction} and Section~\ref{sec:CoNet} elaborate detailed techniques. Section~\ref{sec:experiment} reports experimental results. Section~\ref{sec:relatedwork} reviews related work. Section~\ref{sec:conclusion} concludes this paper.

\vspace{-3mm}
\section{Preliminaries}
\label{sec:pre}

Our problem is defined over an undirected attributed graph $G(V, E, F)$ where $V$ is the set of nodes with a cardinality of $\left|V \right|=n$ and $E\subseteq V\times V$ is the set of edges. 
$F=\{F_1, \cdots, F_n\}$ is the set of node attributes and $F_i$ is the attributes of node $v_i$. Note that each node may have multiple attributes. 
We use $F^d$ to denote the set of distinct attributes.
The community is denoted by $C(V_C, E_C, F_C) $ where $V_C\subseteq V$ and $F_C\subseteq F$. 
For each $e=(v_i, v_j) \in E_C$, $e \in E$ and $v_i, v_j \in V_C$. {\color{black}The query $q=<V_q, F_q>$ consists of the query nodes $V_q$ and query attributes $F_q$.}
$C_q$ is utilized to denote the corresponding community w.r.t. $q$. When the context is clear, we abbreviate $C_q$ as $C$.
The frequently used notations are summarized in Table~\ref{tab:symbol}.

\vspace{1mm}
\noindent \textbf{Graph Modularity.} We use modularity which is a common metric of graph cohesiveness for the candidate subgraph extraction. {\color{black}It is a parameter-free measure~\cite{DBLP:conf/sigmod/KimLCY22} and represents the proportion of edges that belong to a particular group minus the expected proportion if the edges are randomly distributed}. The higher the graph modularity is, the more cohesive the community is. The classical modularity of a community is defined as:

\begin{definition}{(Classical Modularity~\cite{newman2004finding}).} 
Given a graph $G(V, E, F)$ and a community $C(V_C, E_C, F_C) $, the classic modularity is defined as:

\begin{equation}
    \label{equ:classical_modularity}
    \textrm{CM}(G, C)=\frac{1}{2\left| E\right|}(2\left| E_C\right|-\frac{d_C^2}{2\left| E\right|})
\end{equation}
where $d_C$ is the sum of degrees of the nodes in $C$.

\end{definition}

\begin{table}[t]
\centering 
\caption{Symbols and Descriptions}
\vspace{-0.4cm}
\label{tab:symbol}
\begin{tabular}{|p{1.8cm}|p{6.0cm}|}
\hline
\cellcolor{lightgray}\textbf{Notation} & \cellcolor{lightgray}\textbf{Description} \\ \hline
$G(V, E, F)$ & a graph with attributes in node\\ \hline
$C(V_C, E_C, F_C)$ & a community\\ \hline
$q=<V_q, F_q>$ & a query with node set $V_q$ and attribute set $F_q$\\ \hline
$C_q, \tilde{C}_q$ & the ground-truth/estimated community of $q$\\ \hline
$CM(\cdot),DM(\cdot)$ & classical modularity and density modularity\\ \hline
$BM(\cdot)$ & bipartite modularity for bipartite community\\ \hline
$DSM(\cdot)$ & density sketch modularity\\ \hline
\end{tabular}
\vspace{-0.6cm}
\end{table}

{\color{black} When employing classic modularity for CS~\cite{DBLP:conf/sigmod/KimLCY22}, it suffers from the free-rider effect~\cite{wu2015robust} wherein the resulting community may encompass numerous nodes unrelated to the query nodes, and the resolution limit problem~\cite{fortunato2007resolution} that the resultant community may be too large to highlight some important structures.}
\begin{definition}{(Free-rider effect~\cite{wu2015robust}).}
Given a set of query $q$, let $C$ be a community identified based on a goodness function $f$, and $C^*$ be the optimal solution (either local or global). The goodness function is said to be affected by the free-rider effect if $f(C\cup C^*)\geq f(C)$.
\end{definition}

\begin{definition}{(Resolution limit problem~\cite{fortunato2007resolution}).}
Given a graph $G$, query $q$, the objective function $f$, a community constraint $T$, a community $C$ satisfying $T$ and containing all the query $q$, and any community $C'$ satisfying the constraint $T$ such that $C\cup C'$ is connected and $C \cap C' = \varnothing$, the objective function is said to suffer from the resolution limit problem if there exists a community $C'$ such that $C\cup C'$ satisfies the constraint $T$ and $f(C\cup C')\geq f(C)$.
\end{definition}

Note that the free-rider effect is different from the resolution limit problem. The former pertains to finding the most effective solution for detecting the effect, whereas the latter operates under the assumption of independent communities and ensures connectivity between the given communities. To better alleviate the above issues and incorporate modularity for community search, density modularity is proposed in~\cite{DBLP:conf/sigmod/KimLCY22}. 
\begin{definition}{(Density Modularity~\cite{DBLP:conf/sigmod/KimLCY22}).}
Given a graph $G(V, E, F)$ and a community $C(V_C, E_C, F_C) $, the density modularity of $C$ is defined as: 
\begin{equation}
    \label{equ:density_modularity}
    \textrm{DM}(G, C)=\frac{1}{2\left| V_C\right|}(2\left| E_C\right|-\frac{d_C^2}{2\left| E\right|})
\end{equation}
where $d_C$ is the sum of degrees of the nodes in $C$.
\end{definition}

While there exist some other definitions of modularity like generalized modularity density~\cite{guo2023resolution}, recent research~\cite{DBLP:conf/sigmod/KimLCY22} claims that {\color{black}density modularity is one of the most effective forms of modularity for CS}. Therefore, we concentrate on analyzing and comparing classical modularity and density modularity in this paper.

\vspace{1mm}
\noindent \textbf{Graph Neural Networks.}
Modern GNNs follow a strategy of neighborhood aggregation mechanism, where the representation of a node is iteratively updated by aggregating the representations from its neighbors and its previous layer.
\begin{equation}
    \label{equ:pre_GIN}
    h_v^{(k)}=\textrm{M}\left ( h_v^{(k-1)}, \textrm{AGG}\{h_u^{(k-1)}:{u\in N(v)}\}   \right ) 
\end{equation}
where {\color{black} $h_v^{(k)}$ is the representation of node $v$ in layer $k$, $N(v)$ is the neighbors of node $v$}, $\textrm{AGG}$ is the aggregate function to aggregate messages, 
and $\textrm{M}$ is the message propagation function that updates the representation of the node by the aggregated representations and its own representation from the previous layer. Different variants of GNNs have been proposed according to their own method of assigning weights to the neighbors and aggregating information.

\vspace{1mm}
\noindent \textbf{{\color{black} The framework} for Learning-based ACS}. The general process for learning-based ACS includes two steps, i.e., the offline model training and the online query steps. The query set which contains both the query nodes and query attributes, the corresponding ground-truth community, and the graph with attribute information are used as inputs. It first trains the model on the training dataset offline and then utilizes the learned model to predict the test queries online. To ensure connectivity, the constrained BFS is used for community identification~\cite{jiang2022query} that selects the nodes with a score larger than the pre-defined threshold and there exists a path from the node to query while scores of nodes in the path are all larger than the pre-defined threshold.

\vspace{1mm}
\noindent \textbf{Problem Statement}. Given an attributed graph $G(V, E, F)$, and a query $q=⟨V_q, F_q⟩$ where $V_q \subseteq V$ is a set of query nodes and $F_q \subseteq F$ is a set of query attributes, the task of Attributed Community Search (ACS) aims to find a query-dependent community $C_q$, which preserves both structure cohesiveness and attribute homogeneity (i.e., nodes
in the community are densely intra-connected, and the attributes of these nodes are similar).

\section{Overview of \netname}
\label{sec:overview}

\begin{figure}
  \centering
  \includegraphics[width=0.85\linewidth]{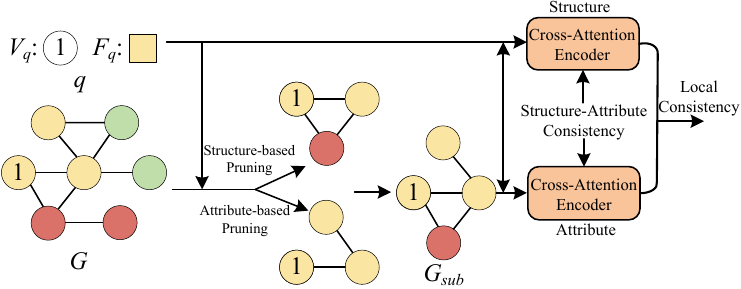}
  \vspace{-3mm}
  \caption{The framework of \netname}
  \label{fig:framework}
\vspace{-6mm}
\end{figure}

In this paper, we present a novel learning-based approach, named \netname, for solving the problem of ACS. The overall framework is illustrated in Figure~\ref{fig:framework}. Given the query set $q$ and the data graph $G$, \netname~first extracts the candidate subgraph from the data graph using the query. The pruning stage contains two branches to select the promising candidates. The first branch is to extract the candidate subgraph considering the structure cohesiveness, while the second branch is to extract the candidate subgraph considering the  semantic homogeneity. Both the structure-based candidates and attribute-based candidates are then combined as one candidate subgraph for the downstream prediction. The candidate subgraph together with the query set are sent to \aggname~to search for the community. The \aggname~comprises two branches, one dedicated to structure and the other to attribute, each utilizing the cross-attention encoder to learn the representations. {\color{black}Two consistency constraints including the structure-attribute consistency constraint and the local consistency constraint are used to guide the training.} After that, \aggname~outputs the predicted community. The detailed technique for candidate subgraph extraction is in Section~\ref{sec:Candidate_Subgraph_Extraction} and the detailed architecture of \aggname~is introduced in Section~\ref{sec:CoNet}.

\section{Candidate Subgraph Extraction}
\label{sec:Candidate_Subgraph_Extraction}

In this section, we introduce details of the candidate subgraph extraction scheme specifically devised for ACS. As previously discussed in Section~\ref{sec:intro}, the current cutting-edge ACS method is trained over the entire graph, thus inherently limiting its efficiency and scalability on large graphs. Here, we outline the desired features that serve as guidelines for developing the candidate subgraph extraction techniques: 1) \textit{Adaptiveness.} An ideal subgraph extraction method ought to be capable of adaptively selecting and determining an appropriate quantity of candidate nodes, considering both the query and the data graph. 2) \textit{Structure-Attribute awareness.} Since ACS aims to identify a structurally cohesive subgraph that upholds attribute homogeneity, the extraction method must pay heed to both the structure and attribute factors to retain as many promising candidate nodes as possible.
Driven by these requirements, we present a modularity-based extraction approach. The proposed method involves a twofold process: firstly, detecting structurally cohesive candidate subgraph; and secondly, engaging in attribute-based pruning. The resultant candidate nodes from both phases are subsequently combined to form the candidate subgraph.

\vspace{-3mm}
\subsection{Structure-based Pruning}
Our extraction scheme is based on modularity, a popular parameter-free metric of cohesiveness. As outlined in Section~\ref{sec:pre}, there exist multiple types of modularity defined for different scenarios. In this paper, our focus is on classical modularity, which is {\color{black}one of the earliest proposed modularities}, and density modularity, which is deemed one of the most powerful forms of modularity for CS, as documented in~\cite{DBLP:conf/sigmod/KimLCY22}. However, classical modularity is known to suffer from the free-rider effect and the resolution limit problem, which may result in the selection of too many loosely-connected nodes. On the other hand, density modularity may impose overly stringent requirements on cohesiveness, which may hinder the exploration of additional promising nodes. Hence, both are unsuitable for candidate subgraph extraction for ACS. To strike a balance and harness the benefits of the above two modularities, we propose the density sketch modularity as follows:

\begin{definition}{(Density Sketch Modularity).}
\label{def:density_generalized_modularity}
Given a graph $G(V, E, F)$, a community $C(V_C, E_C, F_C) $ and a positive real number $\tau \in \mathbb{R}^{+}$, the density sketch modularity is defined as:
\begin{equation}
    \textrm{DSM}(G, C)=\frac{1}{2\left| V_C\right|^{\tau}}(2\left| E_C\right|-\frac{d_C^2}{2\left| E\right|})
\end{equation}
where $d_C$ is the sum of degrees of the nodes in $C$.
\end{definition}

By manipulating the value of $\tau$, we can attain varying levels of cohesiveness granularity.

\begin{equation*}
    \lim \limits_{\tau \to 0} DSM(G,C) = \frac{1}{2}(2\left| E_C\right|-\frac{d_C^2}{2\left| E\right|})
\end{equation*}
When $\tau$ approximates zero, the difference between classical modularity and density sketch modularity is the $\left| E\right|$ in the denominator which is a constant throughout the community within the same data graph. However, we only compare its relative magnitudes within one data graph when using modularity. Hence, density sketch modularity shares the same power as classical modularity when $\tau$ approximates zero.

\begin{equation*}
    \lim \limits_{\tau \to 1} DSM(G,C) = \frac{1}{2\left| V_C\right|}(2\left| E_C\right|-\frac{d_C^2}{2\left| E\right|})
\end{equation*}
When $\tau$ approximates one, density sketch modularity is exactly the density modularity.

As proven below, density sketch modularity also shares two nice properties as density modularity for any varying $\tau \in \mathbb{R}^{+}$.
\begin{lemma}
\label{equ:free_rider}
Whenever density sketch modularity suffers from the free-rider effect, classic modularity suffers from the free-rider effect as well.
\end{lemma}

\begin{proof}
Assume that $C$ is an identified community, and $C^*$ is the optimal solution. If the density sketch modularity suffers from the free-rider effect, we can obtain from the definition that $DSM(G, C\cup C^*)\geq DSM(G, C)$. Putting the definition of \textit{DSM} into the inequality, we can get
$\frac{1}{2\left| V_{C\cup C^*}\right|^{\tau}}(2\left| E_{C\cup C^*}\right|-\frac{d_{C\cup C^*}^2}{2\left| E\right|}) \geq \frac{1}{2\left| V_C\right|^{\tau}}(2\left| E_C\right|-\frac{d_C^2}{2\left| E\right|})$. As  $2\left| V_C\right|^{\tau}$ is larger than zero for $\tau \in \mathbb{R}^{+}$, we can multiply both sides by $2\left| V_C\right|^{\tau}$ and get $\{\frac{\left| V_C\right|}{\left| V_{C\cup C^*}\right|}\}^{\tau}(2\left| E_{C\cup C^*}\right|-\frac{d_{C\cup C^*}^2}{2\left| E\right|}) \geq 2\left| E_C\right|-\frac{d_C^2}{2\left| E\right|}$. As $\{\frac{\left| V_C\right|}{\left| V_{C\cup C^*}\right|}\}^{\tau}$ is always smaller than 1 for $\tau \in \mathbb{R}^{+}$, hence we get the following inequality:
$2\left| E_{C\cup C^*}\right|-\frac{d_{C\cup C^*}^2}{2\left| E\right|} \geq \{\frac{\left| V_C\right|}{\left| V_{C\cup C^*}\right|}\}^{\tau}(2\left| E_{C\cup C^*}\right|-\frac{d_{C\cup C^*}^2}{2\left| E\right|}) \geq 2\left| E_C\right|-\frac{d_C^2}{2\left| E\right|}$. We use the first and third items from the inequality and multiply both sides by $\frac{1}{2\left| E\right|}$. After that we can get the following expression: $ \frac{1}{2\left| E\right|}(2\left| E_{C\cup C^*}\right|-\frac{d_{C\cup C^*}^2}{2\left| E\right|}) \geq \frac{1}{2\left| E\right|}(2\left| E_C\right|-\frac{d_C^2}{2\left| E\right|})$ which is exactly $CM(G, C\cup C^*)\geq CM(G, C)$. Therefore, if density sketch modularity suffers from the free-rider effect, classic modularity also suffers from the free-rider effect as well.
\end{proof}.

\vspace{-4mm}
\begin{lemma}
\label{equ:resolution_limit}
Whenever density sketch modularity suffers from the resolution limit problem, classic modularity suffers from the resolution limit problem as well. 
\end{lemma}

\begin{proof}
    
The proof of LEMMA~\ref{equ:resolution_limit} is quite similar to that of LEMMA~\ref{equ:free_rider}. Assume that $C$ and $C'$ are two communities satisfying $C \cap C' = \varnothing$ and $G[C \cup C']$ being a connected subgraph. And then, we get the proof of LEMMA~\ref{equ:resolution_limit} by replacing the $C^*$ of inequalities in the proof of LEMMA~\ref{equ:free_rider} with $C'$.

\end{proof}
\vspace{-4mm}


{\color{black} When identifying the candidate subgraph, existing works either select a small portion of nodes (like ICS-GNN solely depends on one-hop neighbors of the query node) or use the whole graph (like AQD-GNN). Based on density sketch modularity, we choose the subgraph induced by the $k$-hop neighbors of the query nodes that has the highest modularity value as the candidate subgraph. Note that, in this manner, the candidate subgraph is obtained adaptively, and we do not need to pre-set the value of $k$. In addition, we use $\tau \in [0, 1]$ to control the granularity of the subgraph, and a higher $\tau$ value can produce a more cohesive subgraph. {\color{black}We set $\tau$ as 0.8 by default as suggested by our experimental results in Section \ref{Ablation_Study}.}}

\begin{example}
For the query node 4 in Figure~\ref{fig:case}, its 1-hop subgraph contains nodes 2, 4, and 6. Thus, its modularity is $\frac{1}{2\times3^{0.8}}(2\times 3-\frac{10^2}{2\times14})=0.504$ by \textit{DSM} with $\tau=0.8$. Similarly, we can get its modularity of 2-hop induced subgraph $\frac{1}{2\times7^{0.8}}(2\times9-\frac{23^2}{2\times14})=-0.094$ and the modularity of 3-hop induced subgraph is $\frac{1}{2\times10^{0.8}}(2\times14-\frac{28^2}{2\times14})=0.0$. Hence, the subgraph induced by its 1-hop neighbors is selected as the structure-based candidate subgraph. 
\end{example}

\vspace{-3mm}
\subsection{Attribute-based Pruning}

\begin{figure}[t]
  \centering
  \includegraphics[width=1.00\linewidth]{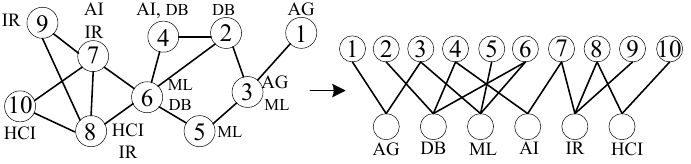}
  \vspace{-7mm}
  \caption{node-attribute bipartite graph}
  \label{fig:fusion_graph}
\vspace{-4mm}
\end{figure}

\begin{algorithm}[t]
\caption{Candidate Subgraph Extraction}
\label{algo:k-hop modularity}
\LinesNumbered
\DontPrintSemicolon
\KwIn{The query $q=<V_q,F_q>$, the attributed graph \color{black}{$G=(V,E,F)$}. }
\KwOut{The candidate subgraph $G_{sub}$}

\tcp*[l]{The structure-based pruning.}

Initialize set $P=V_q$, $Q=V_q$, {\color{black}{$V_{sub}=V_q$}}, $max\_mod=-inf$\;
\While{$\left|Q \right| \textless \left|V \right| $}{
    \For{each $v \in Q$}{
    $P\leftarrow P\cup N(v)$
}
$Q\leftarrow P$; $mod$ = calculate \textit{DSM}(G, Q)\;
\If{$mod \textgreater max\_mod$}{
 $max\_mod \leftarrow mod$; {\color{black}{$V_{sub}=V_{sub}\cup Q$}}
}
}

\tcp*[l]{The attribute-based pruning.}

$BG(V_B=(U,L), E_B) \leftarrow $ construct the bipartite graph.\;
$P, Q \leftarrow $ Query attribute node in $BG$; $max\_mod=-inf$\; 

\While{$\left|Q \right| \textless \left|U \right|+\left|L \right| $}{
    \For{each $v \in Q$}{
    $P\leftarrow P\cup N(v)$
}
$Q\leftarrow P$; $mod$ = calculate \textit{BM}(BG, Q)\;
\If{$mod \textgreater max\_mod$}{
 $max\_mod \leftarrow mod$; {\color{black}{$V_{sub}=V_{sub}\cup Q.U$}}
}
} 
{\color{black}{$G_{sub}\leftarrow $induced subgraph from $V_{sub}$}}\;
\Return $G_{sub}$\;
\end{algorithm}

Attributes play a crucial part in the search for the attributed community. In order to establish the connection between nodes and attributes, we create a node-attribute bipartite graph $BG(V=(U,L), E)$ based on the approach in ~\cite{jiang2022query}. 
The node-attribute bipartite graph contains two types of node sets: the graph node set $U$ and the attribute node set $L$. Each distinct attribute is represented by an attribute node in the $L$ side of the node-attribute bipartite graph, and there is a link between $u\in U$ and $l\in L$ if $l$ is the attribute of node $u$ in the original graph.

\begin{figure}[t]
  \centering
  \includegraphics[width=0.85\linewidth]{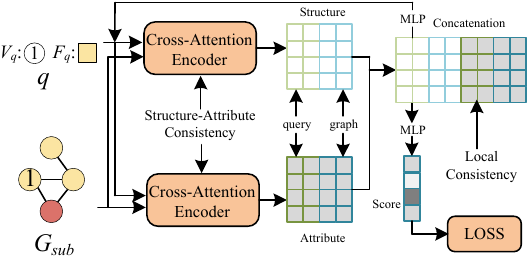}
  \vspace{-5mm}
  \caption{Illustration of \aggname}
  \label{fig:agg_framework}
\vspace{-6mm}
\end{figure}

\begin{example}
An example of the node-attribute bipartite graph   {\color{black} using the graph in Figure~\ref{fig:case}} is depicted in Figure~\ref{fig:fusion_graph}. The original graph contains 10 nodes with 6 distinct attributes. Therefore, there are 10 nodes on the $U$ side and 6 nodes on the $L$ side. The attributes of node 4 contain \textit{AI} and \textit{DB}. Hence, there is an edge between node 4 and node \textit{AI}, and an edge between node 4 and node \textit{DB}.
\end{example}

By utilizing the bipartite graph, nodes in the original graph that are with the same attributes share a common neighbor in the bipartite graph. Additionally, attribute nodes with similar neighborhoods may possess similar semantics. To measure the cohesiveness of the subgraph in the bipartite graph, we resort to bipartite modularity:

\begin{definition} (Bipartite Modularity).
Given a bipartite graph $G(V=(U,L), E)$ and a community $C(V_C=(U_C,L_C), E_C) $, the bipartite modularity is defined as follows~\cite{kim2022abc}:
\begin{equation}
    \label{equ:bipartite_modularity}
     \textrm{BM}(G, C)=\frac{1}{\left| E\right|}(2\left| E_C\right|-\frac{d_C^Ud_C^L}{\left| E\right|})
\end{equation}
where $d_C^U$ is the sum of degrees of the nodes in the $U$ side of $C$ and $d_C^L$ is the sum of degrees of the nodes in the $L$ side of $C$.
\end{definition}

Similar to structure-based pruning, we present an attribute-based pruning approach that leverages the bipartite modularity to identify semantically-similar nodes. Given the query attribute nodes in the bipartite graph, {\color{black} the attribute-based candidate nodes are ascertained based on the subgraph induced by the \ksize-hop neighbors of the query attribute nodes possessing the largest bipartite modularity. Note that, only graph nodes in the induced subgraph are selected since the final output community is a set of graph nodes.}

\vspace{1mm}
\noindent
{\bf The candidate subgraph extraction algorithm.}
The overall modularity-based candidate subgraph extraction method is shown in Algorithm~\ref{algo:k-hop modularity}. The algorithm takes the attributed graph and the query as inputs and produces the candidate subgraph. Initially, it performs structure-based pruning (lines 1 to 7). It first designates the query nodes as the candidate subgraph and the maximum modularity as negative infinity during initialization (line 1). {\color{black} Then, it expands outward hop by hop until all nodes are taken into consideration (lines 2 to 7). For each $k$, if the density sketch modularity of the induced subgraph formed by the $k$-hop neighbors of the query nodes exceeds the previous maximal value, the candidate subgraph is expanded to encompass the $k$-hop neighbors (lines 3 to 7)}. Next, the algorithm performs attribute-based pruning (lines 8 to 15). It begins by creating the node-attribute bipartite graph (line 8) and maps the query attributes into attribute nodes in the bipartite graph (line 9). Lines 10 to 15 are analogous to lines 2 to 7, with the exception that only the nodes on the $U$ side are selected (line 15). 
{\color{black}{At last, the algorithm outputs the candidate subgraph $G_{sub}$ induced from the selected candidates $V_{sub}$ (lines 16 and 17).}}

\vspace{-3mm}
\section{Consistency-aware Net}
\label{sec:CoNet}

\begin{algorithm}[t]
\caption{Forward Propagation of \aggname}
\label{algo:forwardpropagation}
\LinesNumbered
\DontPrintSemicolon
\KwIn{The qeury $q=<V_q,F_q>$, candidate subgraph $G_{sub}$. }
\KwOut{The predicted community $\Tilde{C}_q$.}

$H_{v_q}^{(0)}, H^{\left(s,0\right)}, H_{f_q}^{(0)}, H^{(a,0)}\leftarrow$  feature initialization\;

\For{ $k=0,\cdots, K-1$}{
{\color{black}$X_q, X_k, X_v =  H_{v_q}^{(k)} W_q^{(s,k)}, H^{\left(s,k\right)} W_k^{(s,k)}, H^{\left(s,k\right)} W_v^{(s,k)}$}\;
$ H_{v_q}^{(k+1)} = \textrm{softmax}(\frac{X_qX_k^T}{\sqrt{d}}) X_v$\;

    \For{$v \in V(G_{sub})$}{
        $h_v^{(s, k+1)}={\color{black}\textrm{MLP}^{(s,k)}}\left ( \left ( 1+\epsilon ^{(k)} \right )\cdot h_v^{(s, k)}+ {\textstyle \sum_{v\prime\in N(v)}}h_v\prime^{(s, k)}   \right ) $
    }
{\color{black}$X_q, X_k, X_v =  H_{f_q}^{(k)} W_q^{(a,k)}, H^{\left(a,k\right)} W_k^{(a,k)}, H^{\left(a,k\right)} W_v^{(a,k)}$}\;
$ H_{f_q}^{(k+1)} = \textrm{softmax}(\frac{X_qX_k^T}{\sqrt{d}}) X_v$\;

    \For{$v \in V(G_{sub})$}{
        $h_v^{(a, k+1)}={\color{black}\textrm{MLP}^{(a,k)}}\left ( \left ( 1+\epsilon ^{(k)} \right )\cdot h_v^{(a, k)}+ {\textstyle \sum_{v\prime\in N(v)}}h_v\prime^{(a, k)}   \right ) $
    }
$H^{(s)} = H_{v_q}^{(k+1)} || H^{(s, k+1)}, H^{(a)} = H_{f_q}^{(k+1)} || H^{(a, k+1)}$\;
$H_{v_q}^{(k+1)}= {\color{black}\textrm{MLP}^{(v_q, k)}}(H^{(s)} || H^{(a)})$\;
$H_{f_q}^{(k+1)}={\color{black}\textrm{MLP}^{(f_q, k)}}(H^{(s)} || H^{(a)})$\;
}

$\Tilde{C}_q = \textrm{MLP}(H^{(s)} || H^{(a)})$\;

\Return $\Tilde{C}_q$\;
\end{algorithm}

Based on the obtained substructure {\color{black}$G_{sub}$} and the input query $q=<V_q, F_q>$, a consistency-aware net, namely \aggname, is designed to predict the community. 
The overall architecture is illustrated in Figure~\ref{fig:agg_framework} and Algorithm~\ref{algo:forwardpropagation}. \aggname~incorporates three main components, including 1) Cross-attention encoder, 2) Structure-attribute consistency, and 3) Local consistency. \aggname~ runs for $K$ layers and uses two cross-attention encoders to encode the structure and attribute information of both the query and data graph into the latent space (lines 2 to 13). Both structure and attribute representations are concatenated and fed into an MLP for predicting the community (line 14). 
While learning the representation, two consistency constraints are employed to guide the training. The structure-attribute consistency constraint aims to obtain consistent representations for the structure and for the attribute, while the local consistency constraint aims to achieve an aligned prediction result for neighboring nodes. 
In the following subsections, we describe these main components in detail.

\vspace{-2mm}

\subsection{Feature Initialization}
As the GNN model needs vectorized inputs, we introduce a vectorization technique for the structure iuput and attribute input.

\noindent\textbf{Structure Input.} The query node set $V_q$ is encoded as a one-hot vector {\color{black}$H_{v_q}^{(0)} \in \{0,1\}^{|V_{sub}|}$} where the $i$-th bit equals to one if $v_i \subseteq V_q$. For example, the query node $v_2$ is encoded as $[0, 1, 0, 0, 0, 0, 0, 0, 0, 0]^T$ for the graph in Figure~\ref{fig:case}. As the size of the candidate subgraph is small, the length of the one-hot vector is also small. 

\noindent\textbf{Attribute Input.} The query attribute set $F_q$ is encoded as a one-hot vector {\color{black}$H_{f_q}^{(0)} \in \{0,1\}^{|V_{sub}|}$} where the $i$-th bit equals to one if there exist an attribute $f_j$ while $f_j  \subseteq F_q $ and $f_j  \subseteq F_i $. Each bit indicates the relevance of node attributes and query attributes.  For example, given the query attribute set $\{DB\}$ and graph in Figure~\ref{fig:case}, the query attribute set is encoded as $[0, 1, 0, 1, 0, 1, 0, 0, 0, 0]^T$. The input of the candidate subgraph is the stack of the feature of each node. {\color{black}Note that although different subgraphs of equal length may share an initial query node/attribute representation, their final representation will differ since the GNN model propagates the representation through different edges in different subgraphs.}

\vspace{-2mm}

\subsection{Cross-Attention Encoder}

\begin{figure}[t]
  \centering
  \includegraphics[width=0.75\linewidth]{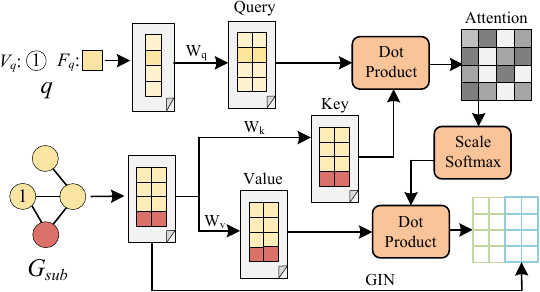}
  \vspace{-3mm}
  \caption{Illustration of Cross Attention Encoder}
  \label{fig:cae}
\vspace{-6mm}
\end{figure}

The query and data graph are not isolated entities, but instead, have a high correlation w.r.t. the resulting community. 
This interaction plays a crucial role in learning query-dependent embeddings. Encoding the two inputs separately would lead to a lack of information exchange between them, resulting in indistinct representations and diminishing the accuracy.
To capture such interaction, we design a cross-attention encoder that leverages a cross-attention mechanism to learn the embeddings. 
The overall illustration of the cross-attention encoder is shown in Figure~\ref{fig:cae}. 

In the retrieval system, elements are stored in a ``key-value'' pair format where the key serves as the identifier for the corresponding value which can be a large file. When a query is submitted, the system compares the query with each key stored in the system. If a key matches the query, the corresponding value is retrieved and returned by the system.
The design of cross-attention follows this architecture of the \textit{query-key-value} retrieval. 

The cross-attention encoder is utilized to encode both structure and attribute information. Specifically, we showcase the utilization of the cross-attention encoder for the encoding of structure information in layer $k$.
The structure inputs consist of the query nodes $H_{v_q}^{(k)}$ and the graph $H^{\left(s,k\right)}$. The query goes through a linear transformation layer which is parameterized by a weight matrix ${\color{black}W_q^{(s,k)}} \in \mathbb{R}^{d_k \times d_{k+1}}$ where $d_k$ is the dimension of the hidden vector of layer $k$. 
Similarly, we project the graph into the latent space by two weight matrices ${\color{black}W_k^{(s,k)}, W_v^{(s,k)}} \in \mathbb{R}^{d_k \times d_{k+1}}$ for key and value. {\color{black}We use superscript "s" for structure-related components and "a" for attribute-related components.}

\begin{equation}
    X_q =  H_{v_q}^{(k)} W_q^{(s,k)}, \ X_k =  H^{\left(s,k\right)} W_k^{(s,k)}, \ X_v = H^{\left(s,k\right)} W_v^{(s,k)}
\end{equation}

Next, to calculate the similarity between the query and key, we perform a dot-product operation, followed by scaling the result by the square root of the dimension, and applying a softmax function to normalize the resulting vector. This produces an attention matrix {\color{black}$X \in \mathbb{R}^{{|V_{sub}|} \times {|V_{sub}|}}$}, where each element $x_{ij} \in X$ represents the correlation between nodes $v_i$ of query and $v_j$ of data graph. Next, we use the dot product of $X$ and $X_v$ to obtain the query representation that combines both information from the query and data graph.
\begin{equation}
    X = \textrm{softmax}(\frac{X_qX_k^T}{{\color{black}\sqrt{d_{k+1}}}}),\ H_{v_q}^{(k+1)} = X X_v
\end{equation}

In addition to encoding the query representation, we also employ GNN to obtain the data graph representation and subsequently concatenate the two representations for community prediction. In this paper, we use Graph Isomorphism Network (GIN)~\cite{DBLP:conf/iclr/XuHLJ19} as the backbone, which has an excellent structure-preserving ability and has been utilized in various graph analytic tasks such as subgraph counting~\cite{wang2022neural, DBLP:conf/icde/WangZQWZL22} and graph classification~\cite{maron2019provably}. The formula of GIN of cross-attention encoder for structure encoding is as follows:

\begin{equation}
    \label{equ:GIN}
    h_v^{(s, k+1)}={\color{black}\textrm{MLP}^{(s,k)}}\left ( \left ( 1+\epsilon ^{(k)} \right )\cdot h_v^{(s, k)}+ {\textstyle \sum_{v\prime\in N(v)}}h_v\prime^{(s, k)}   \right ) 
\end{equation}
where $h_v^{(s, k)} \in H^{(s, k)}$ is the latent representation of node $v$ in layer $k$, $N(v)$ is the neighbor set of $v$ and $\epsilon ^{(k)}$ is a learnable parameter in layer $k$. An MLP is utilized to learn to combine the information from the neighborhood and the previous layer. The result is concatenated with the query representation to obtain the structure representation:


\begin{equation}
    \label{equ:structure representation}
    H^{(s)} = H_{v_q}^{(k+1)} || H^{(s, k+1)}
\end{equation}
where $||$ is the concatenate operation. Similarly, we can get the attribute representation $H^{(a)}$. The encoder runs for $K$ layers, and each layer uses an MLP to pass query information to the next layer:

\begin{equation}
    H_{v_q}^{(k+1)}= {\color{black}\textrm{MLP}^{(v_q, k)}}(H^{(s)} || H^{(a)})
\end{equation}
And similarly, we can get the $H_{f_q}^{(k+1)}$.

The cross-attention encoder is based on the formula of the scaled dot-product attention from the self-attention mechanism~\cite{vaswani2017attention}, but with two key differences: 1) There are two different sequences of vectors are used for the cross-attention encoder, while only one sequence of vectors is used for the self-attention mechanism; 2) The self-attention aims to capture the relationship between itself and data in the datasets, while the cross-attention encoder focuses on capturing the relationship between the query and the data graph.

\subsection{Structure-Attribute Consistency}

The current non-learning-based approaches for ACS consider the structure and attribute separately, whereas AQD-GNN, a learning-based method, fuses these representations using a feature fusion operator. However, AQD-GNN falls short in considering the correlation between query nodes and query attributes, as it only concatenates these representations. 
Although structures and attributes come from two heterogeneous spaces, they are a paired sample to describe a node, and thus they are close-related to each other and should be in close proximity to each other in the latent space. Each node's representation is a point in the latent space, and representations of multiple nodes form a distribution~\cite{gao2021unsupervised}. In this paper, we introduce a structure-attribute consistency constraint that aims to minimize the discrepancy between the distribution of structure representation and the distribution of attribute representation from the perspective of the whole graph. 


There have been many metrics to measure the discrepancy between two distributions, such as the Kullback-Leibler (KL) divergence~\cite{kullback1951information} and the Jensen-Shannon (JS) divergence~\cite{manning1999foundations}. In this paper, we use the Wasserstein distance~\cite{villani2009optimal}. It has a better property since two distributions converge under Wasserstein distance
while failing to exhibit convergence under KL and JS divergences in some cases~\cite{DBLP:conf/iclr/ArjovskyB17, arjovsky2017wasserstein, wang2022neural}. The Wasserstein distance is defined: 

\begin{definition}{(Wasserstein Distance).}
Given random variables $\mu$ and $\nu$ that are subject to probability distributions $\mathbb{P}_s$ and $\mathbb{P}_a$, the Wasserstein-1 distance $W_1$ between distributions $\mathbb{P}_s$ and $\mathbb{P}_a$ is defined:
\label{def:W_distance}
\begin{equation}
    \label{equ:W_distance}
    W_1(\mathbb{P}_s, \mathbb{P}_a) = \inf_{\gamma\in \pi (\mathbb{P}_s, \mathbb{P}_a)} \mathbb{E}_{(\mu, \nu)\sim\gamma}[||\mu-\nu||] 
\end{equation}
where $\pi (\mathbb{P}_s, \mathbb{P}_a)$ denotes the set of all joint distributions whose
marginals are $\mathbb{P}_s$ and $\mathbb{P}_a$ respectively.
\end{definition}

In this definition, $\gamma$ is the “mass” needed to be transported from
$\mu$ to $\nu$ to transform the distributions $\mathbb{P}_s$ into the distribution $\mathbb{P}_a$. And in this case, the element $\gamma_{i,j}$ is the probability that $h_{v_i}^{(s)}\in H^{(s)}$ matches $h_{v_j}^{(a)}\in H^{(a)}$.
As the infimum in Equation~\ref{equ:W_distance} is highly intractable. The Wasserstein-1
distances can be reformulated with the Kantorovich-Rubinstein duality~\cite{gozlan2017kantorovich}:

\begin{equation}
    \label{equ:W_distance_sup}
    W_1(\mathbb{P}_s, \mathbb{P}_a) = \sup_{||f_w||_L \le 1}\mathbb{E}_{\mu \sim \mathbb{P}_s}[f_w(\mu)] -\mathbb{E}_{\nu \sim \mathbb{P}_a}[f_w(\nu)] 
\end{equation}
where $f_w$ satisfies the 1-Lipschitz condition that map the $\mu$, $\nu$ in variable space to the real space $\mathbb{R}$.

By clamping the weights of $f_w$ to a fixed box~\cite{arjovsky2017wasserstein, wang2022neural}, the Wasserstein distance is minimized when $f_w$
is optimized to minimize:

\begin{equation}
\label{equ:wl_distance}
    \mathcal{L}_w(H^{(s)},H^{(a)}) = \sum_{h_v^{(a)}\in H^{(a)}}f_w(h_v^{(a)})-\sum_{h_u^{(s)}\in H^{(s)}}f_w(h_u^{(s)})
\end{equation}

\vspace{-2mm}
\subsection{Local Consistency}

Existing learning-based community search and attributed community search models have been implemented by modeling the problem as a binary node classification task over the entire graph, which can be flawed as nodes within a real-world community are not isolated but rather cohesively connected as a unified module. Hence, solely relying on the gradient signal from the node classification task may not be sufficient. 
To alleviate this issue, we introduce the local consistency for the ACS to enhance the link between nodes in the predicted community in this section.

In particular, the structure representation and the attribute representation are first concatenated before sending to an MLP for computing the final community score:

\begin{equation}
\label{equ:final_MLP}
    H = H^{(s)} || H^{(a)} 
\end{equation}

We then minimize the loss between the self-multiplication of the concatenated matrix and the adjacency matrix as follows:

\begin{equation}
\label{equ:inter-consistency}
    \mathcal{L}_m(H,A) = \left|\left| A-HH^T\right|\right|_F
\end{equation}
where A is the adjacency matrix, and $\left|\left| \cdot\right|\right|_F$ denotes the Frobenius norm of the matrix~\cite{bottcher2008frobenius}.
The auxiliary link prediction objective $\mathcal{L}_m(\cdot)$ captures the idea that neighboring nodes should be predicted together, based on the intuition that if one node belongs to a community, its neighbors are also likely to belong to the same community, and vice versa.

\subsection{Learning Objectives}

There are three learning objectives in \aggname. The first one is in Equation~\ref{equ:wl_distance} that aims to preserve the structure-attribute consistency via minimizing the Wasserstein distance between the distribution of structure and the distribution of attribute. The second one is in Equation~\ref{equ:inter-consistency} which aims to maintain local consistency to enhance the link between nodes in the community. And the third loss uses the binary cross entropy (BCE) which aims to minimize the difference between the predicted community and the ground-truth community. For a query $q_i$, the predicted community score is denoted as $\Tilde{C}_{q_i} \in [0,1]^{|V_{sub}|}$ and the ground-truth community is denoted as $C_{q_i} \in \{0,1\}^{|V_{sub}|}$,  the BCE loss is defined as:


\begin{equation}
\label{equ:bce}
    \mathcal{L}_b(\Tilde{C}_{q_i},C_{q_i}) = \sum_{j=1}^{|V_{sub}|} -\left(C_{q_i, j}log(\Tilde{C}_{q_i, j})+(1-C_{q_i, j})log(1-\Tilde{C}_{q_i, j})\right)
\end{equation}
where $C_{q_i,j}$ is the $j$-th bit of $C_{q_i}$.

The task of ACS is then modeled as multi-task learning to take the above three losses into account together. The overall loss function is defined as:

\begin{equation}
\label{equ:loss_func}
    \mathcal{L} = \mathcal{L}_b+\alpha\mathcal{L}_w+\beta\mathcal{L}_m
\end{equation}
where $\alpha, \beta \in [0,1]$ are the coefficients to balance the above three loss functions. {\color{black} Note that $\mathcal{L}_w$ and $\mathcal{L}_m$ are two unsupervised losses that do not require the ground-truth labels. This property allows $\mathcal{L}_w$ and $\mathcal{L}_m$ to be easily adaptable to various scenarios, such as incomplete or poor ground-truth data, thereby enhancing the generalization ability and robustness of \textit{ALICE}. }
\subsection{Analysis and Discussion}

We now analyze the expressive power and the structure-preserving ability of the cross-attention encoder. We then further discuss the time complexity of \netname. 

\vspace{1mm}
\noindent\textbf{Expressive Power and Structure-Preserving Ability.} Here we prove that the proposed graph neural network of \aggname~is as powerful as the Weisfeiler-Lehman (WL) isomorphism test. 
\begin{lemma}
    \label{lemma:GIN_power}
There exist parameters for $K$-layered GINs such that, for any positive integer $K$, if the degrees of nodes are bounded by a constant and the size of node features is finite, and for any graphs $g_1$ and $g_2$, if the 1-WL algorithm outputs that $g_1$ and $g_2$ are not isomorphic within $K$ rounds, then the embeddings of $g_1$ and $g_2$ computed by the GIN are distinct.
\end{lemma}

The proof of LEMMA~\ref{lemma:GIN_power} can be found in ~\cite{DBLP:conf/iclr/XuHLJ19}.

\begin{lemma}
    \label{lemma:SAN_power}
    If the 1-WL algorithm outputs that $g_1$ and $g_2$ are not isomorphic within $K$ rounds, the embeddings of $g_1$ and $g_2$ computed by the cross-attention encoder are different.
\end{lemma}

\begin{proof}
The output of the cross-attention encoder concatenates two sources of embedding,i.e., the query $H_q$ and data graph $H_g$, and the GIN is used to encode the data graph. Given $g_1$ and $g_2$, and its outputs are $H_{1}=H_{q_1} || H_{g_1}$ and $H_{2}=H_{q_2} || H_{g_2}$. If $g_1$ and $g_2$ are "non-isomorphic" within $K$ round, $H_{g_1}$ and $H_{g_2}$ which are the output of $g_1$ and $g_2$ by GIN should be different by LEMMA~\ref{lemma:GIN_power}. Hence $H_{1}$ and $H_{2}$ must be different since two different vectors concatenating any vectors will result in two different vectors. Therefore, the cross-attention encoder is as powerful as 1-WL test.
\end{proof}


\vspace{-2mm}
\noindent\textbf{Complexity Analysis.} The time complexity of \netname~consists of the time cost of candidate subgraph extraction and the time cost for \aggname. 
The time complexity of candidate subgraph extraction is $O(|E|+2\times|V|\times|F^d|)$ since it needs to construct the node-attribute bipartite graph and propagates through each edge to find the \ksize-hop neighbors in the data graph and the bipartite graph.
\aggname~ needs to run for $K$ layers and is trained for $t$ epochs, and \aggname~ is applied in the candidate subgraph $G_{sub}=(V_{sub}, E_{sub})$. 
The time complexity of the projection of three matrices is $O(3\times|V_{sub}|\times d^2)$ where $d$ is the maximum latent dimension. The dot product of query and key takes $O(|V_{sub}|^2\times d)$ and the dot product of attention and value also takes $O(|V_{sub}|^2\times d)$. 
The time complexity of GIN applied in the candidate subgraph is $O(|E_{sub}|)$~\cite{wang2022neural}. 
Generally, $|F_d|\approx |V_{sub}|$, hence we assume the time expended for attribute encoding closely approximates that of structure encoding, and the time required for each layer is close.
Therefore, the overall time complexity of \aggname~is 
$O(t\times K\times 2\times (3\times |V_{sub}|\times d^2 + 2\times|V_{sub}|^2\times d + |E_{sub}|))$.





\vspace{-1mm}
\section{Experimental Evaluation}
\label{sec:experiment}

In this section, we evaluate the performance of \netname~ compared with a variety of existing solutions over 11 real-world benchmark datasets with a maximum of 65 million nodes and 1.8 billion edges. 

\vspace{-3mm}
\subsection{Dataset Description}

\begin{table}
\centering
\caption{Statistics of the datasets}
\vspace{-4mm}
\label{tab:dataset}
\scalebox{0.85}{
\begin{tabular}{c||c|c|c|c}
\hline

\cellcolor{lightgray}{Dataset}  & \cellcolor{lightgray}{{$\left| {V} \right|$}} &\cellcolor{lightgray} {{$\left| {E} \right|$}} &\cellcolor{lightgray} {{$| {F^d} |$}} &\cellcolor{lightgray} {{$ N_c $}} \\ \hline
\textit{Texas} & 187   & 279 & 1703 & 5   \\ \hline
\textit{Cornell} & 195   & 285 & 1703 & 5  \\ \hline
\textit{Washt} & 230   & 392 & 1703 & 5 \\ \hline
\textit{Wiscs} & 265   & 469 & 1703 & 5  \\ \hline
\textit{Cora} & 2708   & 5429  & 1433 & 7  \\ \hline
\textit{Citeseer} & 3312   & 4715 & 3703 & 6  \\ \hline
\textit{Google+} & 7856   & 321,268 & 2024  & 91  \\ \hline
\textit{PubMed} & 19,717   & 44,324 & 500 & 3  \\ \hline
\textit{Reddit} & 232,965   & 47,396,905  & 602 & 41  \\ \hline
\textit{Orkut}     & 3,072,627 & 117,185,083  & 1000 & 5000 \\ \hline

\textit{Friendster }  & 65,608,366   & 1,806,067,135   & 1000 & 5000\\ \bottomrule

\end{tabular}}
\vspace{-4mm}
\end{table}

We use 11 public datasets following the previous research~\cite{jiang2022query,gao2021ics,fang2016effective,huang2017attribute} to conduct the experiments. 
The statistics information is summarized in Table~\ref{tab:dataset} where {\color{black} $\left| {V} \right|$ is the number of nodes, $\left| {E} \right|$ denotes the number of edges, $| {F^d} |$ is the number of distinct attributes, and }$N_c$ is the number of communities in the graph. The first nine datasets, including \textit{Texas}, \textit{Cornell}, \textit{Washington (Washt)}, \textit{Wisconsin (Wiscs)}, \textit{Cora}, \textit{Citeseer}, \textit{Google+}, \textit{Pubmed}, and \textit{Reddit}, are attributed graphs with ground-truth communities. To further test the scalability and efficiency of the approaches over large graphs, we add two non-attributed datasets, \textit{Orkut} and \textit{Friendster}, that contain 5000 top-quality ground-truth communities. We generate an attribute pool consisting of $| {F^d} | = 1000$ different attributes for these two graphs following the processing phase in~\cite{huang2017attribute}.

\vspace{-3mm}
\subsection{Experimental Setup}
\vspace{1mm}
\noindent\textbf{Baseline:} We use three baselines for ACS including: 1) ACQ~\cite{fang2016effective}, which is a non-learning $k$-core-based model; 2) ATC~\cite{huang2017attribute}, which is a non-learning $k$-truss-based model; 3) AQD-GNN~\cite{jiang2022query}, which is a learning-based model with feature fusion. We also test the performance of \netname~compared with ICS-GNN~\cite{gao2021ics}, which is a GNN-based interactive community search model. Two baselines are also used for the comparison of non-attributed community search including 1) CTC~\cite{DBLP:journals/pvldb/HuangLYC15}, which is a $k$-truss-based community search model; 2) \ksize-ECC~\cite{chang2015index}, which models communities as $k$-edge connected components.

\vspace{1mm}
\noindent\textbf{Query Setting:} {\color{black}We categorize all ground-truth communities into three distinct groups: training communities, validation communities and test communities. The ratio of these groups is approximately 5:1:4. This process serves to evaluate the performance when the model is exposed to previously unseen communities.} And then, we generate 150 pairs of training data as $\mathcal{D}_{train}=\{q_i, C_{q_i}\}_{i=1}^{150}$. Each query  $q_i=<V_{q_i}, F_{q_i}>$ contains the query node set $V_{q_i}$ and the query attribute set $F_{q_i}$. $C_{q_i}$ is the corresponding ground-truth community of $q_i$ in the training communities. We then generate 100 pairs of validation queries and 100 pairs of test queries. {\color{black} We only split those datasets that have larger than 10 ground-truth communities (i.e., \textit{Google+, Reddit, Orkut, and Friendster}) to avoid insufficient information on ground-truth communities}.
The training data is utilized to train our model, the validation data is used to select the optimal threshold of the predicted score to determine the community, and the test data is employed to evaluate the performance. We randomly select $1\sim 3$ nodes from the ground-truth community as the query nodes of each query. In addition, we use the following three mechanisms to generate the query attributes.
\begin{itemize}[leftmargin=15pt, topsep=1pt]
\item{} \textbf{Empty attribute query (EmA).} We set the attribute query set empty $F_{q_i}=\varnothing$, and query in the EmA is $q_i=<V_{q_i}, \varnothing>$.
\item{} \textbf{Attribute from communities (AFC).} We first select the 5 most common attributes in ground-truth communities and then randomly select one of these attributes as the query attribute. Query in the AFC query set is $q_i=<V_{q_i}, F_{q_i}^c>$.
\item{} \textbf{Attribute from the query node (AFN).} We directly use the attribute from the query nodes as the query attribute. Query in the AFN query set is $q_i=<V_{q_i}, F_{q_i}^n>$.
\end{itemize}

\vspace{1mm}
\noindent\textbf{Metrics:} {\color{black}Combining metrics used in existing works\cite{jiang2022query, fang2016effective, huang2017attribute}, we use three widely used metrics to evaluate the quality of the found communities, including F1-score~\cite{jiang2022query, huang2017attribute}, average degree (Avg.d)~\cite{fang2016effective}, and the Community pair-wise Jaccard (CPJ)~\cite{fang2016effective}. We evaluate the found communities within the candidate subgraph. As the primary objective of ACS is to identify a community that really matches the expectations, we mainly focus on the metric of F1-score which measures the alignment between the found communities and the ground-truth communities.} Given the ground-truth community set denoted as $\mathcal{C}=\{C_{q_1},\cdots,C_{q_t}\}$ and the predicted community sets denoted as $\Tilde{\mathcal{{C}}}=\{\Tilde{C}_{q_1},\cdots,\Tilde{C}_{q_t}\}$, the F1-score which is based on precision and recall is defined as follows. Here, $C_{q_i}$ and $\Tilde{C}_{q_i}$ are the ground-truth and predicted community vectors for query $q_i$. 

\begin{equation*}
    pre(\mathcal{C}, \Tilde{\mathcal{C}})=\frac{\sum_{i=1}^{t}\sum_{j}C_{q_{i},j}\cdot \Tilde{C}_{q_{i},j}}{\sum_{i=1}^{t}\sum_{j}\Tilde{C}_{q_{i},j}},  
    rec(\mathcal{C}, \Tilde{\mathcal{C}})=\frac{\sum_{i=1}^{t}\sum_{j}C_{q_{i},j}\cdot \Tilde{C}_{q_{i},j}}
    {\sum_{i=1}^{t}\sum_{j}{C}_{q_{i},j}}
\end{equation*}

\begin{equation*}
    F1(\Tilde{\mathcal{C}}, \mathcal{{C}})=\frac{2\cdot pre(\mathcal{C}, \Tilde{\mathcal{C}})\cdot  rec(\mathcal{C}, \Tilde{\mathcal{C}})}{ pre(\mathcal{C}, \Tilde{\mathcal{C}})+rec(\mathcal{C}, \Tilde{\mathcal{C}})}
\end{equation*}

{\color{black}
Moreover, we use Avg.d of nodes in the community to measure the structure cohesiveness and the CPJ to measure the attribute cohesiveness as in ~\cite{fang2016effective}. The definitions of Avg.d and CPJ are given:
\begin{equation*}
    Avg.d(\Tilde{\mathcal{C}})=\frac{1}{t}\sum_{i=1}^{t}\frac{1}{|\Tilde{C}_{q_{i}}|}\sum_{v\in \Tilde{C}_{q_{i}}} d(v)
\end{equation*}
\begin{equation*}
    CPJ(\Tilde{\mathcal{C}})=\frac{1}{t}\sum_{i=1}^{t}\frac{1}{|\Tilde{C}_{q_{i}}|^2}\sum_{u\in \Tilde{C}_{q_{i}}} \sum_{v\in \Tilde{C}_{q_{i}}}\frac{|F_u \cap F_v|}{|F_u \cup F_v|}
\end{equation*}
Where $d(v)$ is the degree of node $v$ in $\Tilde{C}_{q_{i}}$ and $F_v$ is the attribute set of node $v$. Attributes that share the same label are considered equivalent.
Note that, a higher value of F1-score or Avg.d or CPJ indicates the higher quality of the identified community.
}

\vspace{1mm}
\noindent\textbf{Implementation Details:} The latent dimension is set as 128. The model is trained for 300 epochs with early stopping. The loss balance coefficients $\alpha, \beta$ in Equation~\ref{equ:loss_func} are both set as 0.1. The $\gamma$ of the density sketch modularity is set as 0.8. The candidate subgraph size is set as 1000 and runs for 20 rounds in ICS-GNN. We conduct experiments using $1\sim 3$ query nodes and AFN as the query attribute by default. The dropout rate is set as 0.45. We set $k$ of the baselines to 4 by default. Adam optimizer with a decaying learning rate is employed to train the model. We conduct our experiments on a machine with Intel(R) Xeon(R) Gold 6248R CPU, 512GB memory, and Nvidia A5000 (GPU).

\subsection{Effectiveness Evaluation}
\label{Effectiveness}

Due to AQD-GNN running out of memory during the training on \textit{Reddit/Orkut/Friendster} and ATC/CTC needing more than 7 days on \textit{Friendster}, we omit these results.

\noindent\textbf{Exp-1: Attributed Community Search.}
We first evaluate our algorithm for the ACS. Figure~\ref{fig:acsresult} reports the result. We present two query scenarios: the one-node query, which employs a single query node, and the multi-node query, which utilizes 1 to 3 nodes as query nodes. Note that ACQ can only use one query node as input while ATC, AQD-GNN, and \netname~can support multiple query nodes. 
Therefore, we use ACQ as the non-learning-based baseline and AQD-GNN as the learning-based baseline for the one-node query in Figure~\ref{fig:acsresult} (a-c). We use ATC as the non-learning-based baseline and AQD-GNN as the learning-based baseline for the multi-node query in Figure~\ref{fig:acsresult} (d-f). When measured by F1-score, we observe that learning-based methods, including AQD-GNN and \netname, exhibit superior performance compared with non-learning methods under both the one-node query and the multi-node query {\color{black}with an average improvement of 41.75\% and 54.50\% in F1-score, respectively.} Overall, \netname~ shows the best performance under both cases with an average F1-score $10.18\%$ improvement compared with AQD-GNN using AFN as the query attributes. 

{\color{black}In terms of Avg.d, non-learning methods possess an inherent advantage since they adopt structure cohesiveness as the objective. Across various queries, non-learning methods consistently tend to incorporate high-degree nodes into the community to promote structure cohesiveness, which is also a contributing factor to the low F1-score. In the context of learning-based methods, \textit{ALICE} outperforms AQD-GNN. The performance stems from the local consistency that aligns the predictions of adjacent nodes, thereby enhancing the connection of nodes within a community. Additionally, the modularity-based pruning prioritizes cohesive subgraphs. When considering CPJ, we find that learning-based approaches exhibit superior performance as they can better approximate the relevance of attributes. Among 11 datasets, \textit{ALICE} outperforms baselines in 8 datasets, which validates the effectiveness of \textit{ALICE}.}

\begin{figure*}
\subfigbottomskip=-2pt 
\subfigcapskip=-6pt 
    \subfigure[\vspace{-3mm}{F1-score of ACS over one-node query}]{ 
        \centering
        \includegraphics[width=0.320\textwidth]{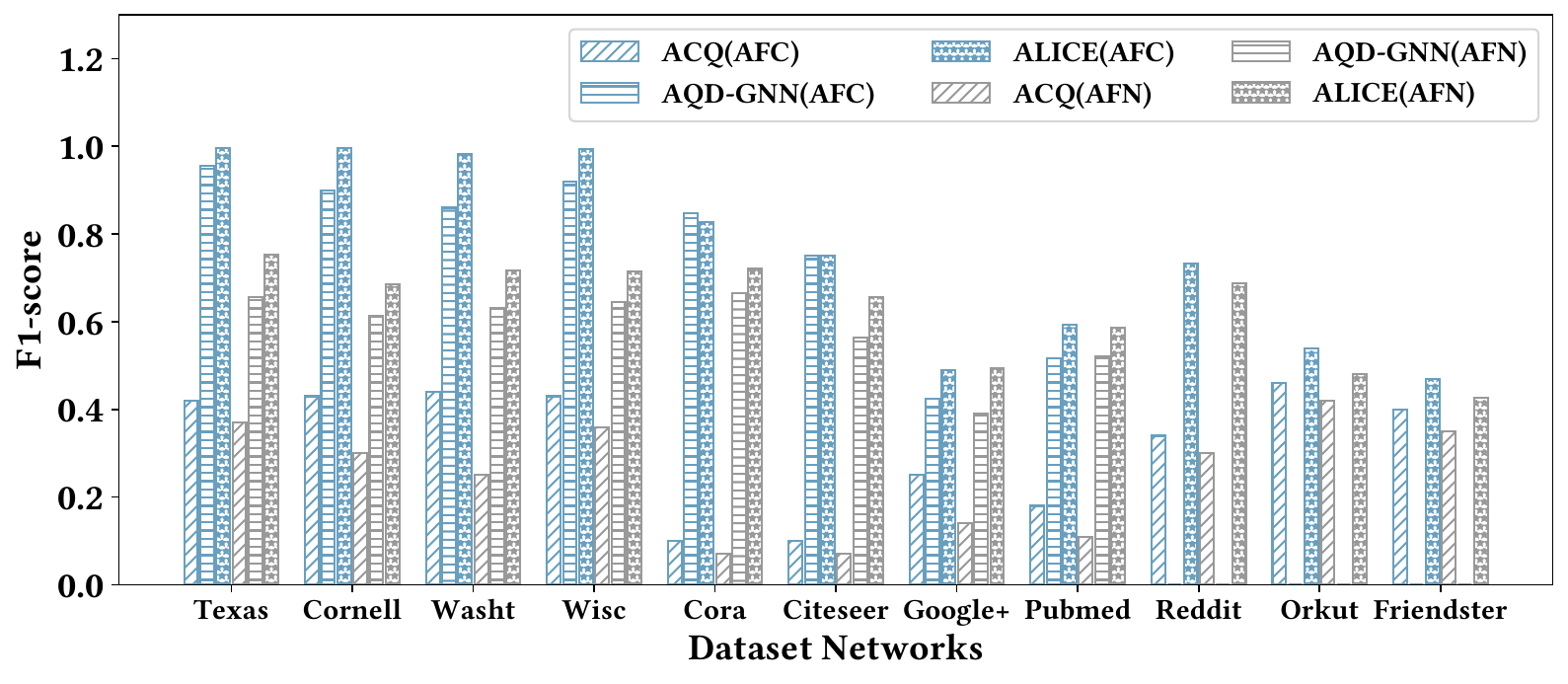}
    \vspace{-3mm}
    }
    \subfigure[\vspace{-3mm}{Avg.d of ACS over one-node query}]{ 
        \centering
        \includegraphics[width=0.320\textwidth]{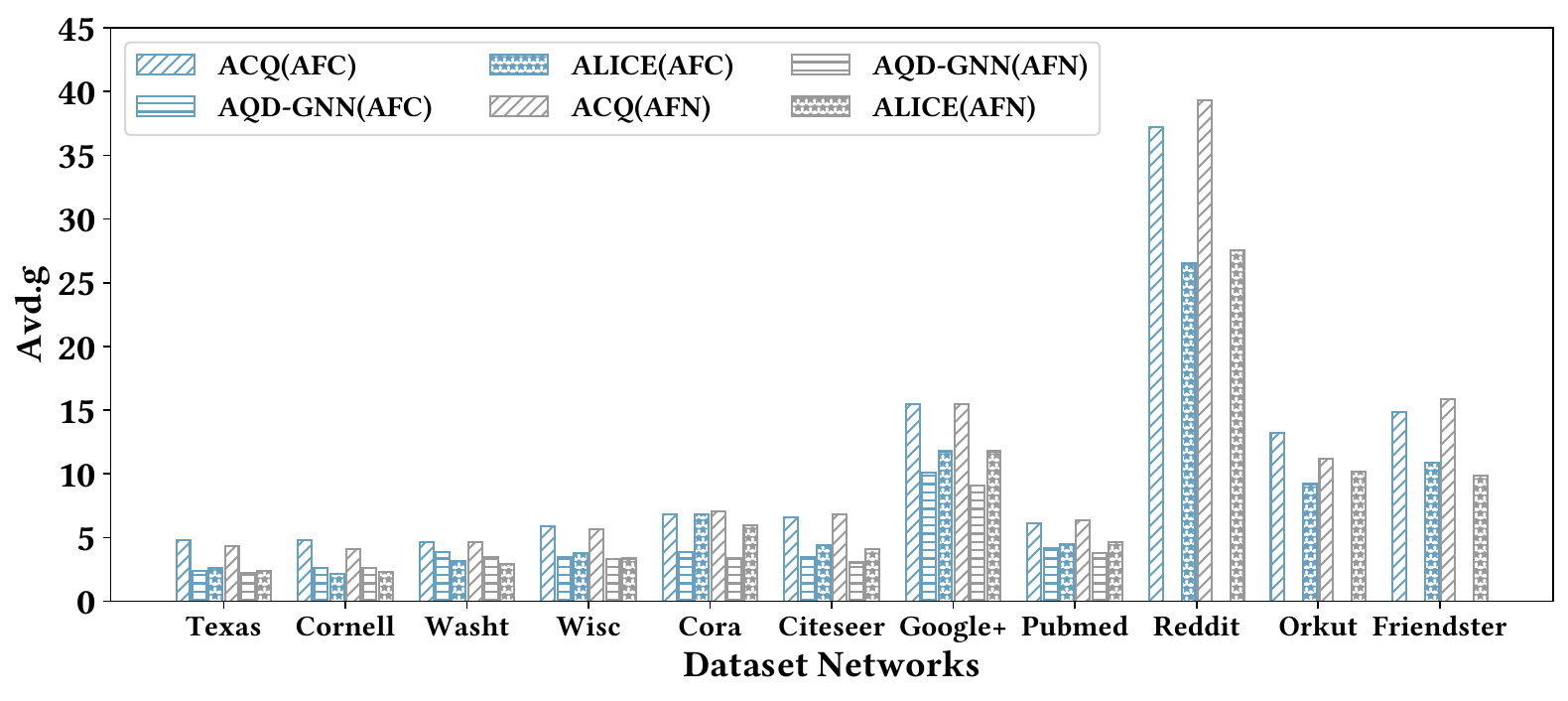}
    \vspace{-3mm}
    }
    \subfigure[\vspace{-3mm}{CPJ of ACS over one-node query}]{ 
        \centering
        \includegraphics[width=0.320\textwidth]{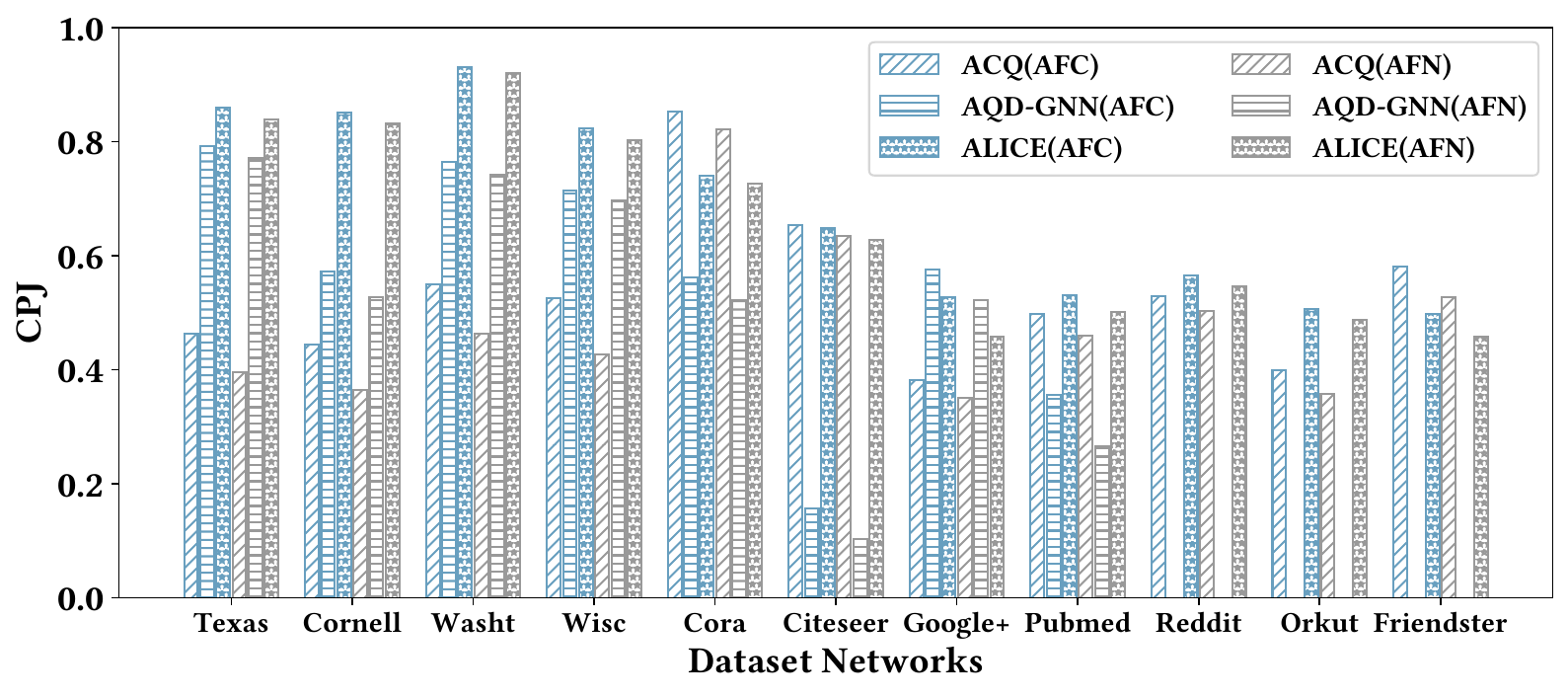}
    \vspace{-3mm}
    }
    \subfigure[{F1-score of ACS over multi-node query}]{
        \includegraphics[width=0.320\textwidth]{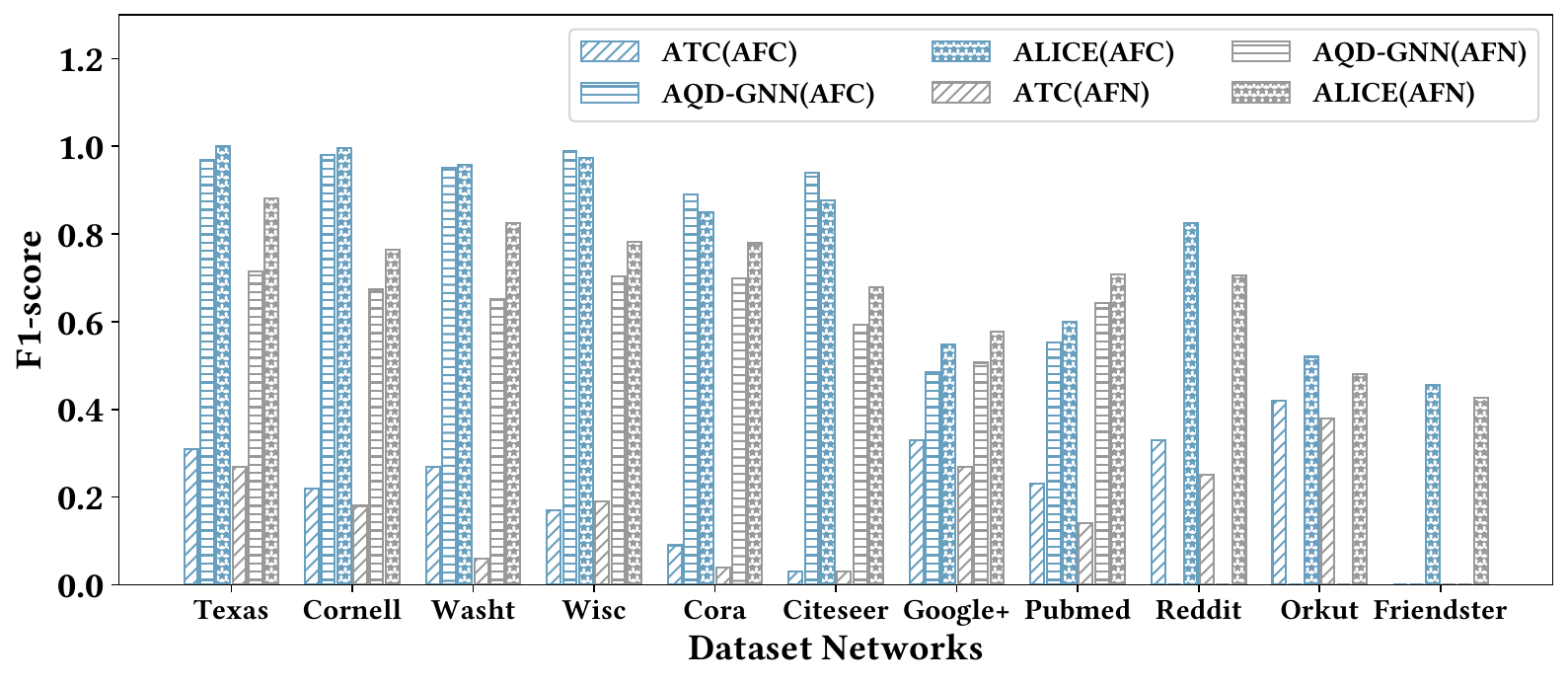} 
    \vspace{-3mm}
    }
    \subfigure[{Avg.d of ACS over multi-node query}]{
        \includegraphics[width=0.320\textwidth]{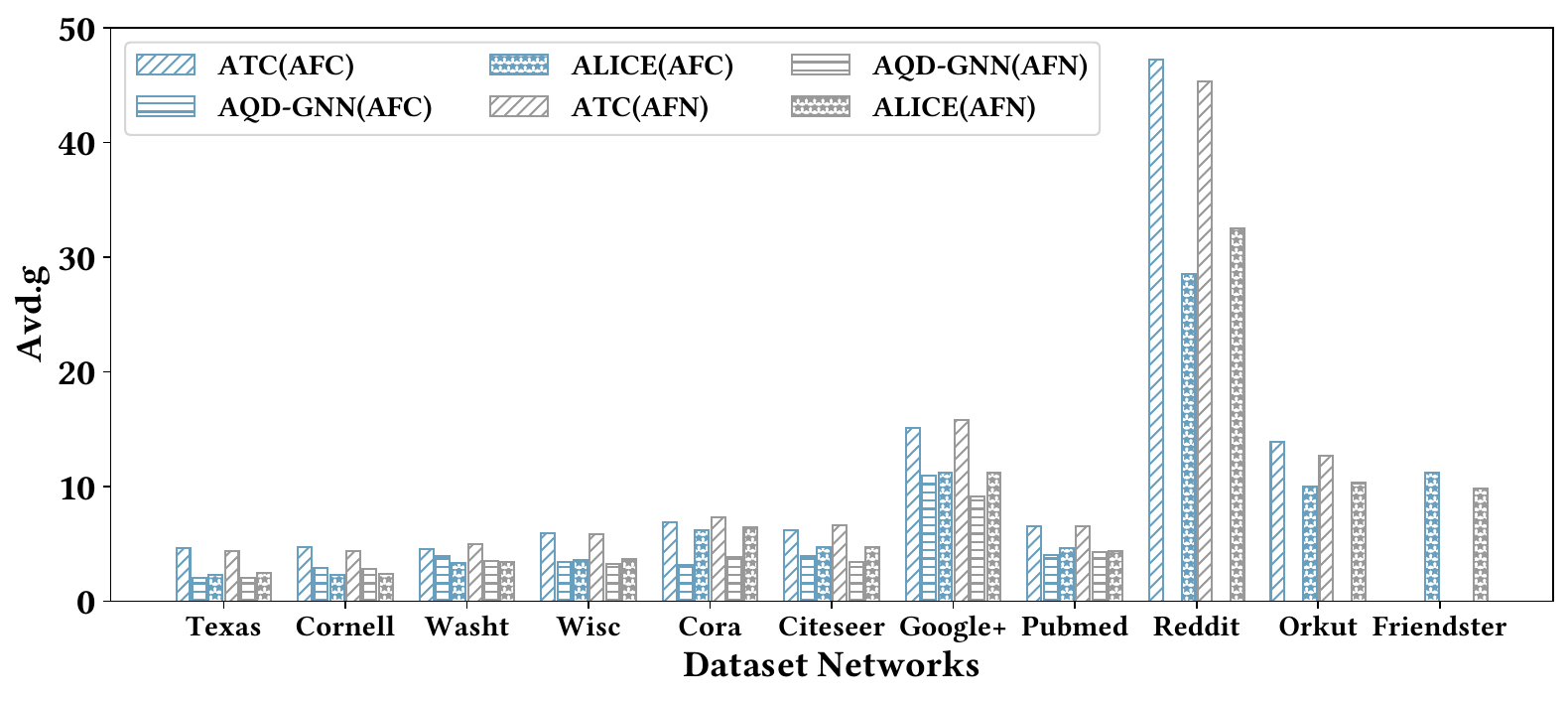} 
    \vspace{-3mm}
    }
    \subfigure[{CPJ of ACS over multi-node query}]{
        \includegraphics[width=0.320\textwidth]{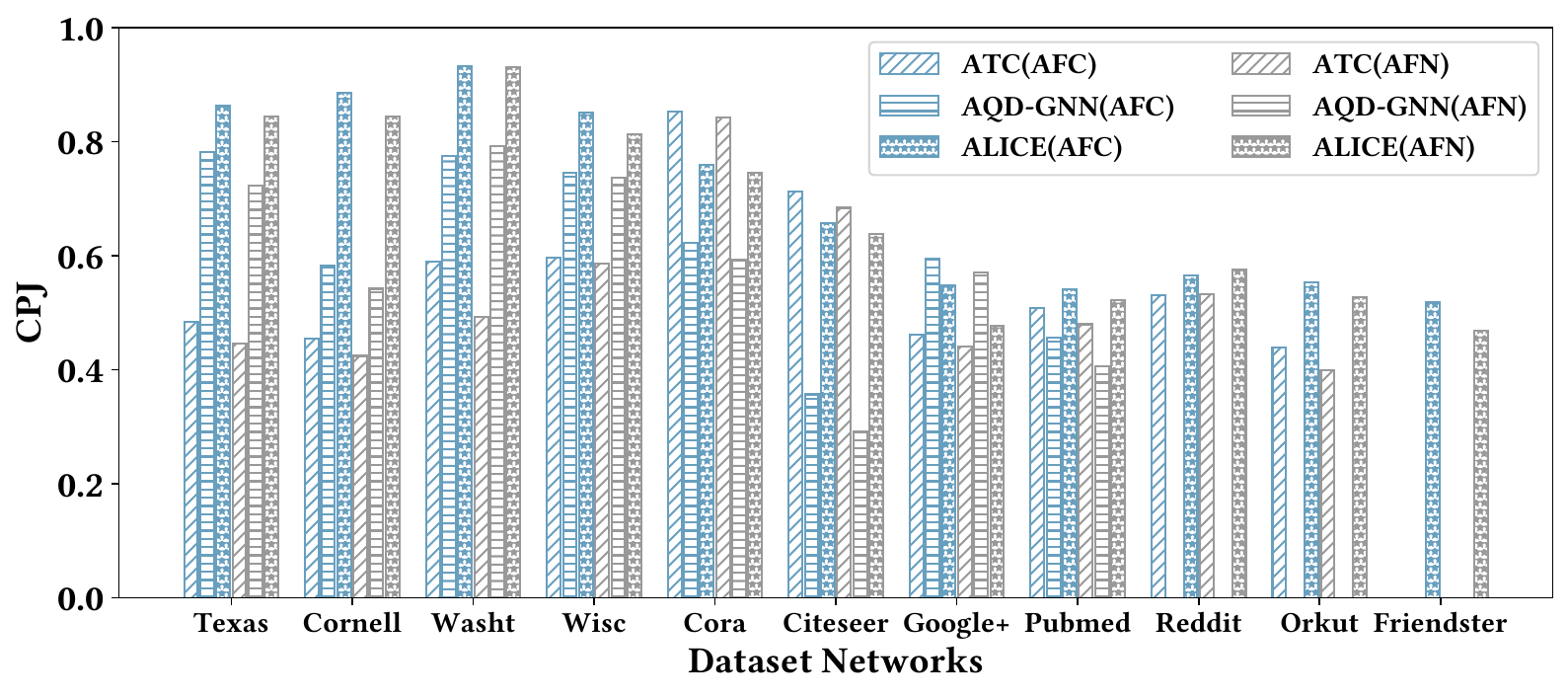} 
    \vspace{-3mm}
    }
\vspace{-4mm}
    \caption{Results on attributed community search}
    \label{fig:acsresult}
\vspace{-3mm}
\end{figure*}

\vspace{1mm}
\noindent\textbf{Exp-2: Non-Attributed Community Search.}
We further test the performance of \netname~for non-attributed community search. \ksize-ECC \cite{chang2015index} and CTC \cite{DBLP:journals/pvldb/HuangLYC15} are used as non-learning-based baselines. Both AQD-GNN and \netname~use the EmA setting for non-attributed community search where the query attribute is set as empty. The overall result is illustrated in Figure~\ref{fig:nac_result} (a). As depicted in the figures, the learning-based approach outperforms non-learning methods significantly. {\color{black}Among all the methods, \netname~ exhibits the highest performance across the datasets with an average improvement of 40.03\% when compared with $k$-ECC in F1-score and of 6.74\% when compared with AQD-GNN in F1-score.} 

\begin{figure*}
\vspace{-2mm}
\subfigcapskip=-6pt 
    \subfigure[{Results on non-acs}]{ 
        \centering
        \includegraphics[width=0.32\textwidth]{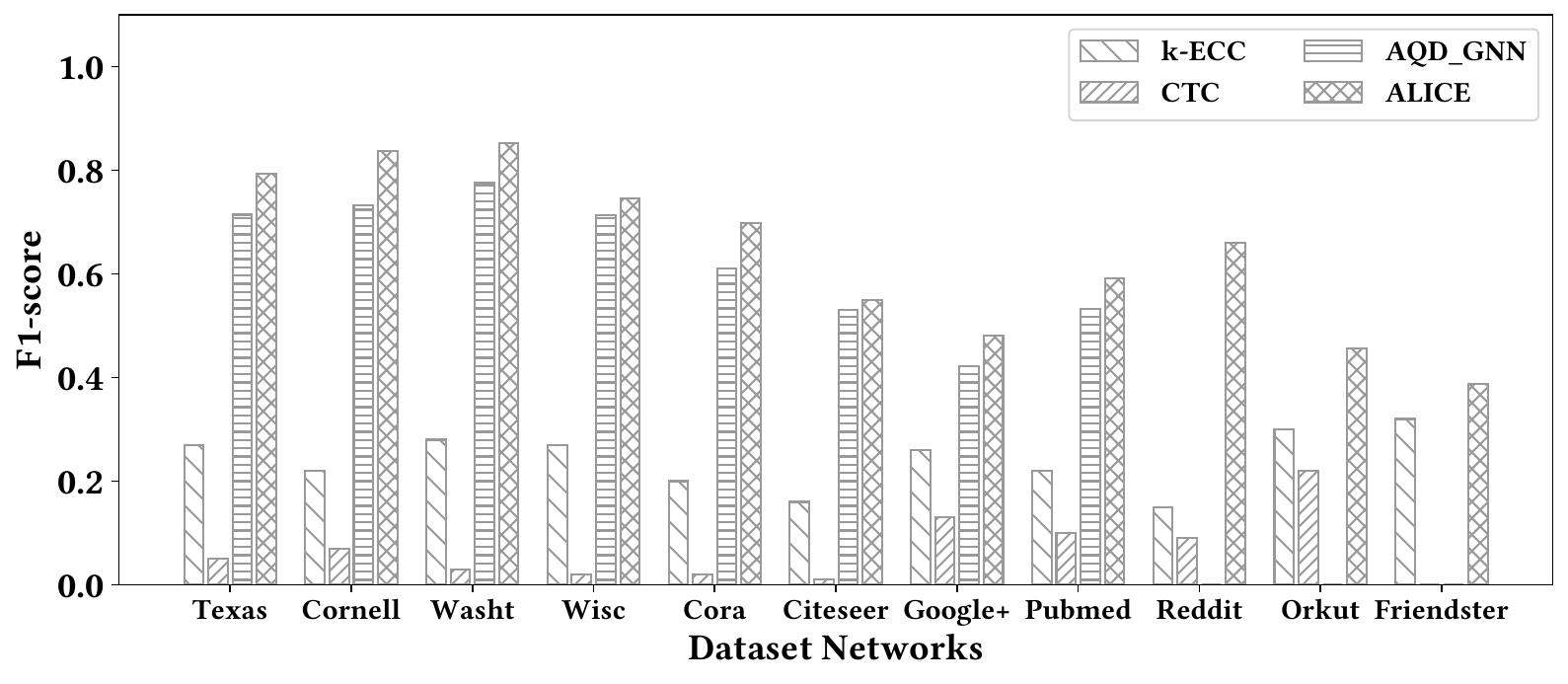}
    \vspace{-3mm}
    }
    \subfigure[{Results on ics}]{
        \centering
        \includegraphics[width=0.32\textwidth]{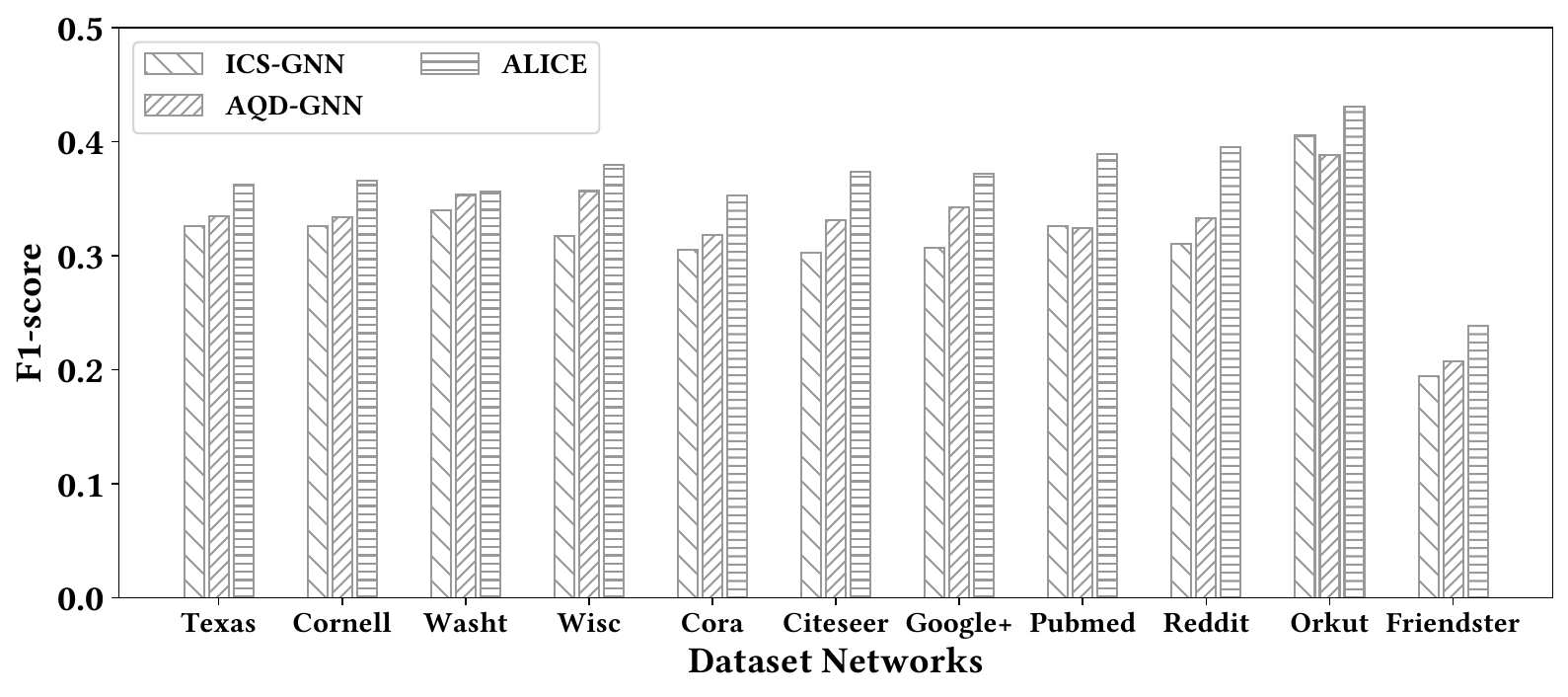} 
    \vspace{-3mm}
    }
    \subfigure[{Results with incomplete ground-truth}]{
        \centering
        \includegraphics[width=0.32\textwidth]{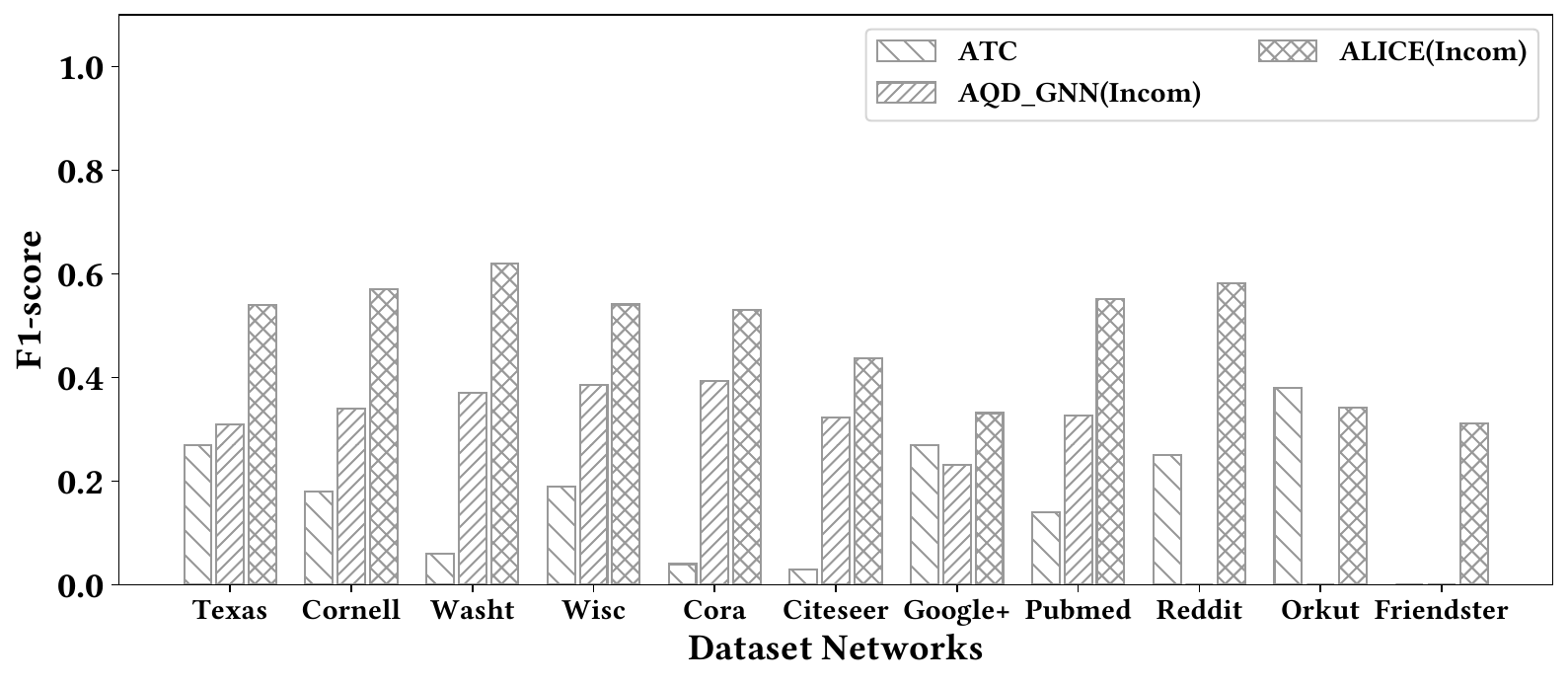} 
    \vspace{-3mm}
    }
\vspace{-6mm}
\caption{Community search with other settings}
\label{fig:nac_result}
\vspace{-3mm}
\end{figure*}

\vspace{1mm}
\noindent\textbf{Exp-3: Interactive Community Search.}
Interactive community search (ICS) is recently proposed in ICS-GNN~\cite{gao2021ics}.
ICS-GNN generates an answer community in response to a query in multiple rounds, and the community can be refined through user feedback. ICS-GNN follows a process of identifying a candidate subgraph, learning node embeddings through a Vanilla GCN model, and using a BFS-based algorithm to select a community of size $k$ with the highest GNN scores. In this part, we replace the Vanilla GCN model with AQD-GNN and \netname~to evaluate the performance of ICS. The results of three models are depicted in Figure~\ref{fig:nac_result} (b). From the figure, we find that {\color{black}\netname~consistently shows the best performance. }

\vspace{1mm}
\noindent\textbf{Exp-4: ACS under incomplete gound-truth.}
{\color{black}This section investigates the robustness of learning-based techniques when dealing with incomplete ground-truth information. Figure~\ref{fig:nac_result} (c) illustrates the results. Specifically, for each training pair $\{q_i, C_{q_i}\}$, we randomly mask 50\% of the nodes in $C_{q_i}$ and rely solely on the remaining half as the ground-truth to train the model. We use AQD-GNN and ATC as baselines as they support multi-node query. The figure shows that \netname~exhibits a notable level of effectiveness, with an average F1-score improvement of 16.05\% and 30.34\% compared to AQD-GNN and ATC, respectively. This is attributable to the modularity-based pruning and the local consistency that put structure cohesiveness prior to the community. Unlike AQD-GNN, which relies solely on the supervised signal $\mathcal{L}_b$, the proposed \netname~ incorporates two new unsupervised signals, $\mathcal{L}_w$ and $\mathcal{L}_b$, to enhance the robustness of the model under incomplete ground-truth data.}

\subsection{Efficiency Evaluation}
\label{Effectiveness}

\noindent\textbf{Exp-5: Evaluation of Different Stages.}
Since the learning-based methods have demonstrated superior performance compared to other baseline methods. Hence, we further assess the efficiency of ICS-GNN, AQD-GNN and \netname. {\color{black} Table~\ref{tab:efficiency_evaluation} presents two phases of time, including train time and query time, both of which are crucial for efficiency. We report data preparation time plus the model train (query) time for each phase. Data preparation involves loading the data, initializing features, and extracting candidate subgraphs. Note that, ICS-GNN needs to train one model for one query, thus we report the total time needed for one query in the query time phase. }
The algorithms are trained using a dataset comprising 150 samples and are subsequently queried using one single data sample. "---" indicates that the algorithm is out of memory or needs more than 7 days during the evaluation. "***" indicates that one cell is not applicable for one model.
The table shows that AQD-GNN is faster than \netname~in the extremely small graph with only hundreds of nodes while \netname~performs much more efficiently for large-scale graphs. Moreover, \netname~ can finish training on \textit{Reddit}, \textit{Orkut} and \textit{Friendster}, whereas ADQ-GNN cannot. The reason is that AQD-GNN uses the whole graph as the input and trains the model. 
For the medium and large graphs, the pruning process can significantly reduce the training size and avoid the issue of being out of memory. The query time of both AQD-GNN and \netname~ are consistent and are much smaller than the preparation and training time. {\color{black}Moreover, ICS-GNN needs much more query time as it needs to train one model for one input query while training a model is a time-intensive task.}

\begin{table*}
\caption{Efficiency evaluation on different datasets (in seconds)}
 \vspace{-4mm}
 
 \centering \scalebox{0.7}{
 \begin{tabular}{|l|c|c|c|c|c|c|c|c|c|c|c|l}\toprule
    Method &{\textit{Texas}}  & {\textit{Cornell}}  & {\textit{Washt}}  & {\textit{Wisc}} 
    & {\textit{Cora}} & {\textit{Citeseer}} & {\textit{Google+}}
    & {\textit{Pubmed}} & {\textit{Reddit}}  & {\textit{Orkut}}
    & {\textit{Frienster}} 
    
                 \\\midrule

  ICS-GNN (Train)& ***                    & ***                      & ***                     & ***                     & ***    & ***  & ***  & ***  & ***  & ***  & ***  
 \\                 
 AQD-GNN (Train) &   2.2+233                     & 2.1+234                 & 2.5+239                     & 2.9+232                  & 64.1+2214 & 59.3+4390                    & 834.6+10035                      & 3171.8+37059 & ---    & ---       & ---                     
 \\
 ALICE (Train)& 2.6+344              & 2.5+381                    & 3.8+332                     & 1.8+324                   & 16.32+509  & 59.8+1239               & 189.8+3256     & 123.5+4317 & 8681+1107  & 2594.8+2224 & 65415.6+1244         
 
    \\
ICS-GNN (Query)& 20.5                    & 25.1                    & 27.4                 & 28.6                   & 167.7    & 124.3  & 627.6  & 112.3  & 1034.7  & 1540.8  & 24253.7 
 \\
 AQD-GNN (Query)&   0.015+0.0021                     & 0.014+0.0020                   & 0.017+0.0022                    & 0.019+0.0020                    & 0.427+0.0026 & 0.395+0.0019                    & 5.564+0.0019                      & 21.14+0.0019 & ---       & ---       & ---                      
 \\
    ALICE (Query) & 0.017+0.0053                     & 0.017+0.0045                     & 0.025 + 0.0044                   & 0.014 +0.0050                    & 0.104+0.0041    & 0.398+0.0047                    & 1.26+0.0053       &  0.823+0.0058  & 5.78+0.0052  & 17.29+0.0045  & 436.1+0.0048             
 
    \\\bottomrule
 \end{tabular}}
 \begin{flushleft}
        \footnotesize $(1):$ We report preparation time $+$ train (query) time; $(2):$ --- indicates out of memory or not finished within 7 days; $(3):$  *** indicates this cell not applicable to this model. 
    \end{flushleft}
 
 \label{tab:efficiency_evaluation}
\end{table*}

\vspace{1mm}
\noindent\textbf{Exp-6: Scalability Evaluation.}
In this part, we evaluate the scalability of AQD-GNN and \netname~ in Figure~\ref{fig: Scalability}. The total time consists of the data preparation time, the model training time, and the query time. The result of the relation between total time and node number is shown on the left of Figure~\ref{fig: Scalability}. The relation between total time and edge number is shown on the right of Figure~\ref{fig: Scalability}. From the figure, we find that the time cost of \netname~grows at a much slower rate than AQD-GNN as the number of nodes or edges increases, which confirms the high scalability of our design.

\begin{figure}
 \vspace{-2mm}
    \centering
    \includegraphics[width=0.22\textwidth]{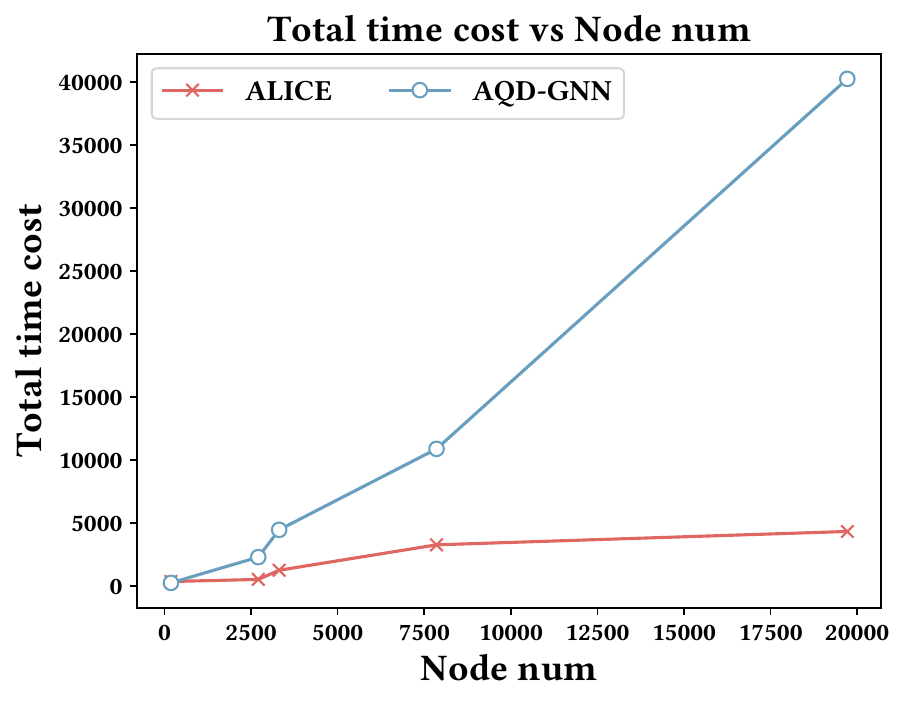} 
     \vspace{0.5em}
    \includegraphics[width=0.22\textwidth]{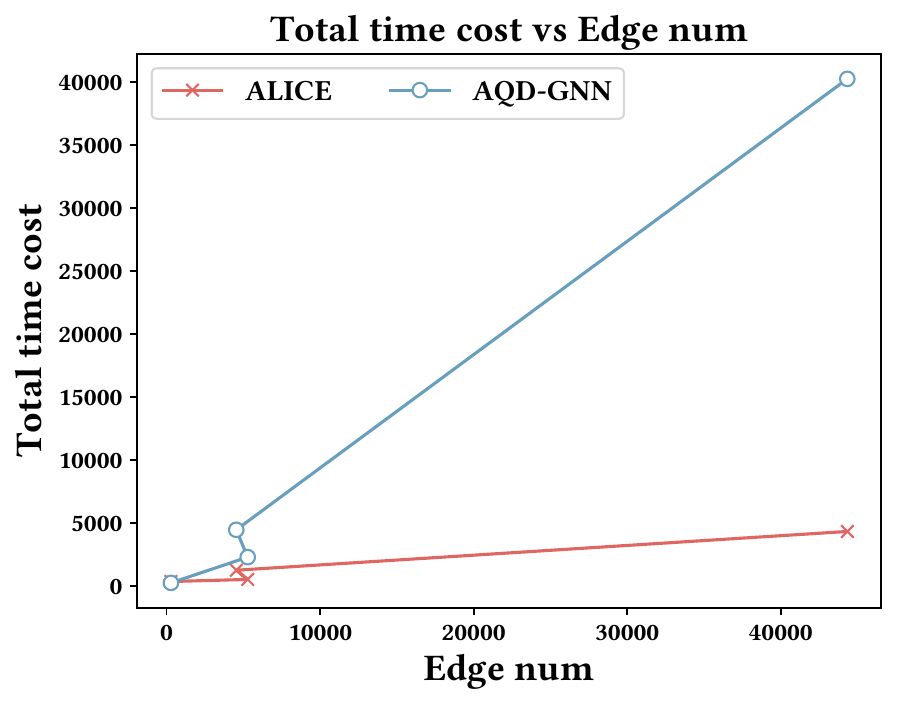}  
    \vspace{-7mm}
    \caption{Scalability evaluation}
    \label{fig: Scalability}
    \vspace{-8mm}
\end{figure}

\begin{figure*}
\vspace{-2mm}
    \includegraphics[width=0.90\textwidth]{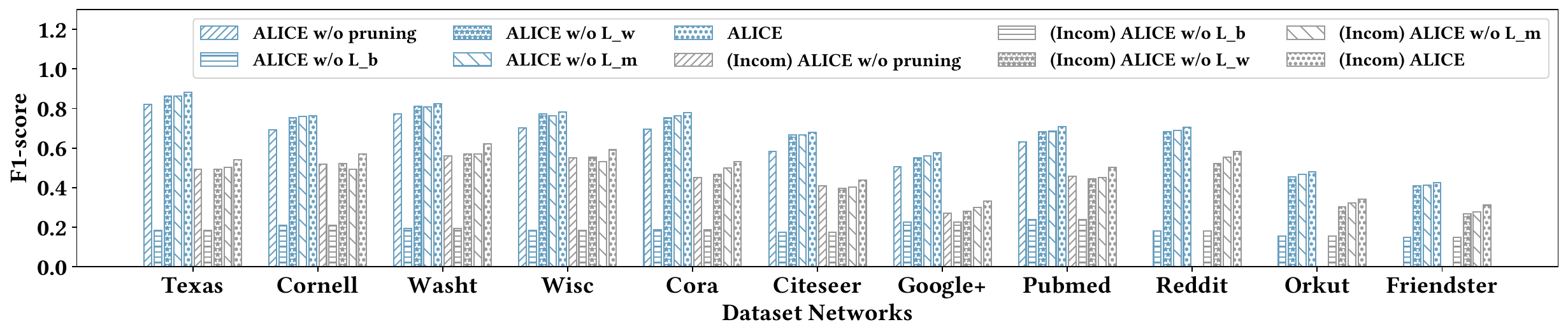}
\vspace{-5mm}
\caption{Ablation study}
\label{fig:ablation_study}
\vspace{-3mm}
\end{figure*}

\begin{figure*}
\vspace{-2mm}
\subfigcapskip=-6pt 
    \subfigure[{Accuracy}]{ 
        \centering
        \includegraphics[width=0.44\textwidth]{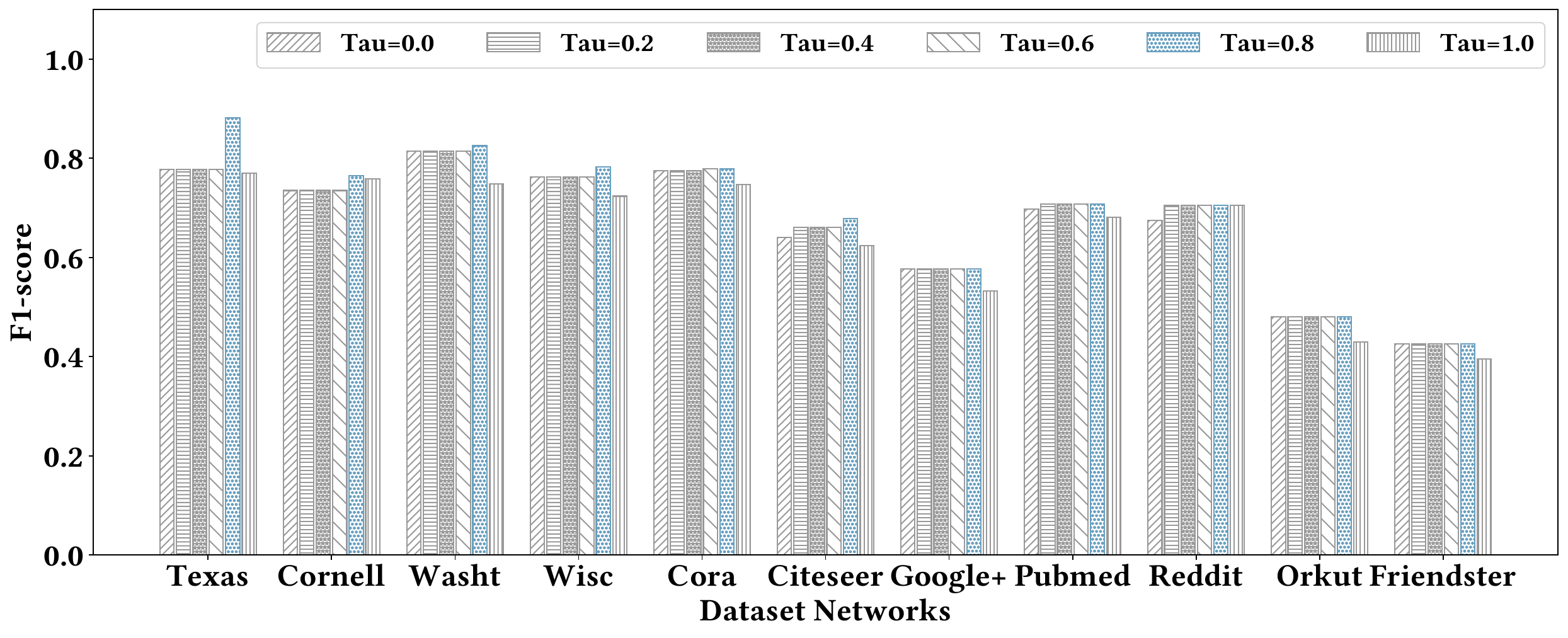}
    \vspace{-4mm}
    }
    \subfigure[{Subgraph size}]{
        \centering
        \includegraphics[width=0.44\textwidth]{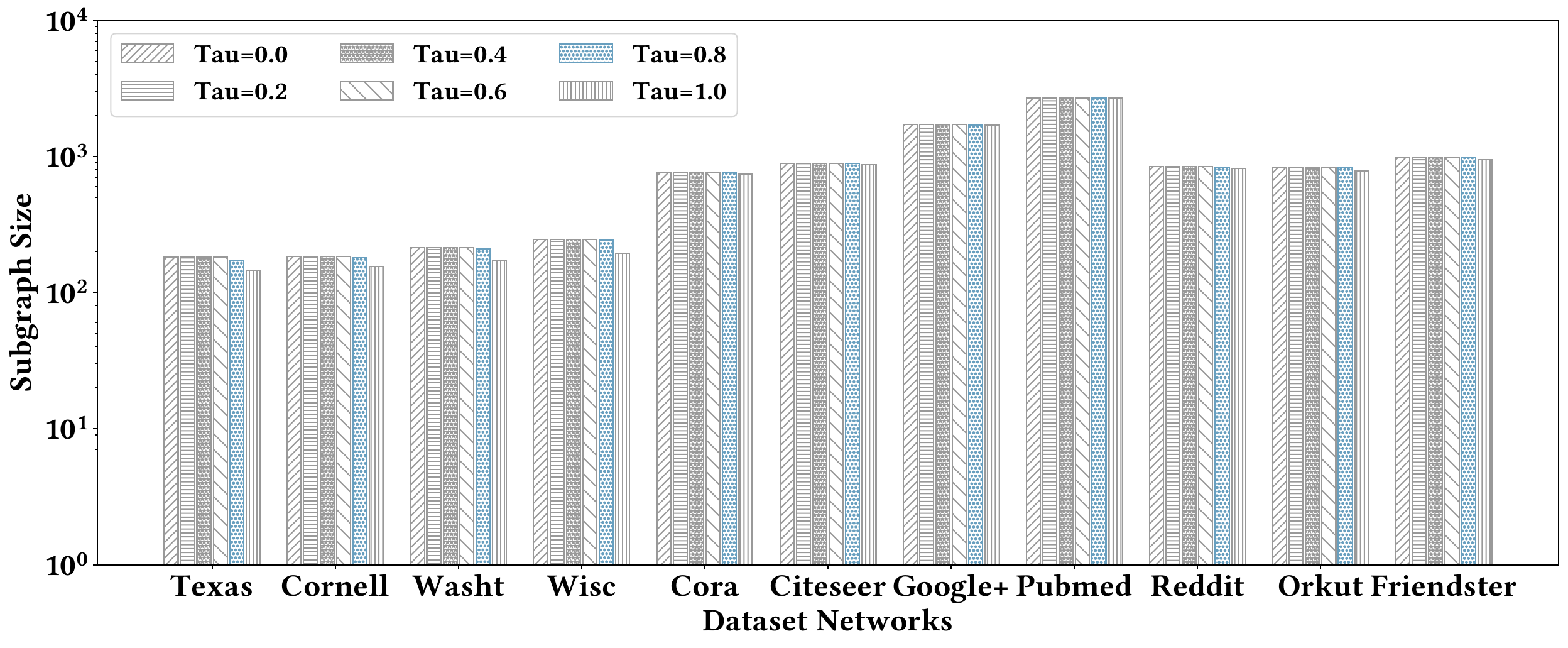} 
    \vspace{-4mm}
    }
\vspace{-7mm}
    \caption{Comparison of different modularities}
    \label{fig:modu_comparation}
\vspace{-6mm}
\end{figure*}

\vspace{-4mm}
\subsection{Ablation Study}
\label{Ablation_Study}
\noindent\textbf{Exp-7: Different Components.} 
In this part, we conduct experiments over different variants of \netname~to verify the effectiveness of each component in the proposed model. The overall results are demonstrated in Figure~\ref{fig:ablation_study}. The first variant (denoted by w/o Pruning) is trained on the whole graph to test the influence of candidate subgraph extraction. The second variant (denoted by w/o $\mathcal{L}_b$) is trained without the loss $\mathcal{L}_b$. The third variant (denoted by w/o $\mathcal{L}_w$) is trained without the loss $\mathcal{L}_w$ to check the effectiveness of the structure-attribute consistency.  The fourth variant (denoted by w/o $\mathcal{L}_m$) is trained without the loss $\mathcal{L}_m$ to confirm the effectiveness of the local consistency. We also test the effectiveness of these components under incomplete ground-truth (denoted by Incom) as Exp-4. 
As shown in the figure, we can find that all four components can improve the accuracy of ACS.
The prediction of the community is most heavily influenced by $\mathcal{L}_b$  since it supplies ground-truth information. The pruning technique can effectively reduce the search space, and hence be helpful to increase both the effectiveness and efficiency of the model. {\color{black} Additionally, it is observed that $\mathcal{L}_w$ and $\mathcal{L}_m$ can enhance the F1-score by 2.71\% and 2.16\% in average respectively under complete ground-truth labels and improve performance by 4.57\% and 4.03\% respectively under incomplete ground-truth information. The results validate the effectiveness of $\mathcal{L}_w$ and $\mathcal{L}_m$, especially under incomplete ground-truth labels. }

\vspace{1mm}
\noindent\textbf{Exp-8: Different Modularity Definitions.}
{\color{black} In this section, we test the model performance with density sketch modularity of varying $\tau$. Note that different $\tau$ will lead to different modularity definitions, e.g, $\tau=0$ leads to classical modularity and $\tau=1$ leads to density modularity as analyzed in Section~\ref{sec:Candidate_Subgraph_Extraction}.} The result is reported in Figure~\ref{fig:modu_comparation}, with the accuracy in Figure~\ref{fig:modu_comparation}(a) and the subgraph size comparison in Figure~\ref{fig:modu_comparation}(b). The figure shows that (1) density sketch modularity with $\tau=0.8$ has a better F1-score than other choices on the test datasets; (2) the subgraph size remains the same or decreases as $\tau$ increases, which aligns with our analysis in Section~\ref{sec:Candidate_Subgraph_Extraction}.

\vspace{-6mm}
\subsection{Hyper-parameter Sensitivity}
\label{subsec:Sensitivity}
\vspace{-1mm}
\noindent\textbf{Exp-9: Varying $\alpha$ and $\beta$.} {\color{black}In this section, we explore the effect of $\alpha$ and $\beta$ to the performance of \netname. We use two datasets including \textit{Texas} and \textit{Cornell}. The results are illustrated in Figure~\ref{fig:alpha_beta_heatmap}. The figure demonstrates that parameters of small but positive values (i.e., 0.1\textasciitilde0.5) outperform other settings. Specifically, we set both $\alpha$ and $\beta$ varying from 0 to 1 and report the F1-score at an interval of 0.1. The results show that (1) $\alpha$ and $\beta$ in the range from 0.1 to 0.5 outperform other settings on the test datasets; (2) $\alpha$ and $\beta$ of small but positive values outperform those of zero. This further validates the effectiveness of $\mathcal{L}_w$ and $\mathcal{L}_m$.}

\vspace{1mm}
\noindent\textbf{Exp-10: Training Epoch.}
In this part, we explore the impact of the training epoch on the accuracy of ACS. We train the model  from 1 to 300 epochs and evaluate the accuracy every 20 epochs. Figure~\ref{fig:sensitivity}(a) presents the results. The figure shows that the performance of \netname~improves considerably as the training epoch increases. After 200 epochs, the rate of improvement decelerates, eventually leading the model to attain a stable state.

\vspace{1mm}
\noindent\textbf{Exp-11: Training Loss.}
In this part, we investigate the loss during training. Figure~\ref{fig:sensitivity} (b) reports the loss of \netname~ over different datasets. In the figure, it is evident that the loss diminishes significantly at the onset, decreasing by approximately $95\%$ following 20 epochs. Thereafter, the loss progressively converges towards 0.

\vspace{1mm}
\noindent\textbf{Exp-12: Training Set Size.}
We also study the sensitivity of the number of samples used in the training phase. We increase the size from 50 to 300 and test the performance of the model every 50 samples.
The result is reported in Figure~\ref{fig:sensitivity} (c). It can be observed that the accuracy increases with the sample size ranging from 50 to 150. After the sample size reaches 150, the accuracy stabilizes.

\vspace{1mm}
\noindent\textbf{Exp-13: Validation Set Size.}
We test the accuracy with varying numbers of samples used in the validation phase. The sample size was increased from 50 to 300, and the accuracy is reported every 50 samples. The results are presented in Figure~\ref{fig:sensitivity}(d). The figure demonstrates that there is a consistent level of accuracy across different sample sizes of validation.

\begin{figure}
\vspace{-2mm}
    \subfigcapskip=-6pt 
    \subfigure[{Texas}]{ 
        \centering
        \includegraphics[width=0.225\textwidth]{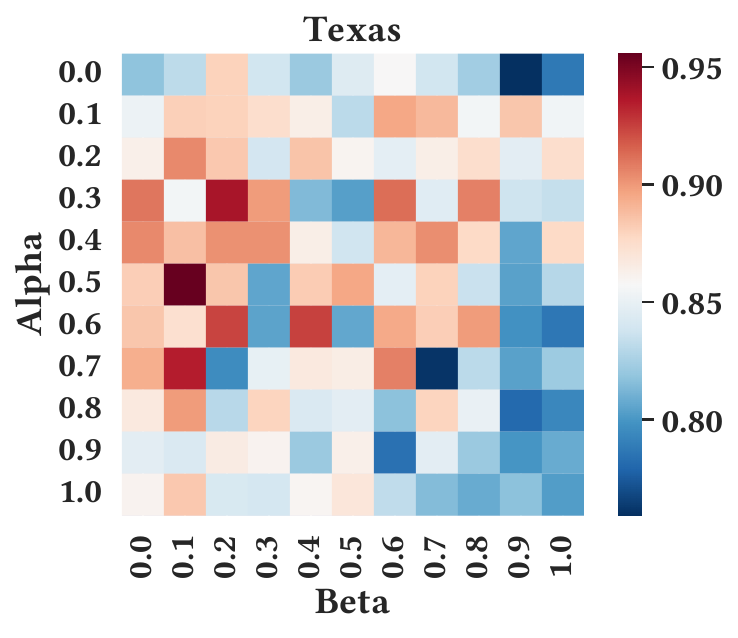}
        \captionsetup{skip=10pt}
      }
    \vspace{-3mm}
    \subfigure[{Cornell}]{
        \centering
        \includegraphics[width=0.225\textwidth]{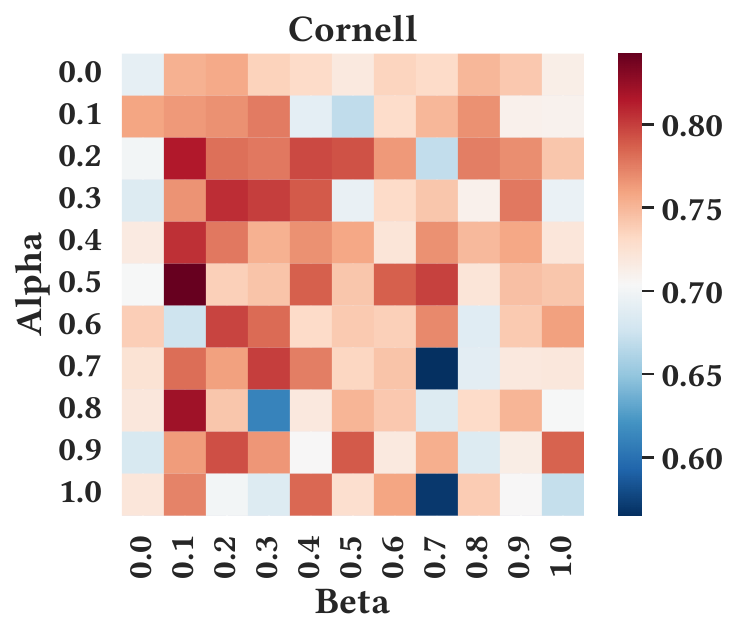} 
    }
    \vspace{-3mm}
    \caption{Results on varying $\alpha$ and $\beta$}
    \label{fig:alpha_beta_heatmap}
\vspace{-3mm}
\end{figure}

\begin{figure}
\vspace{-4mm}
\subfigcapskip=-6pt 
    \subfigure[{Varying epoch number}]{ 
        \centering
        \includegraphics[width=0.225\textwidth]{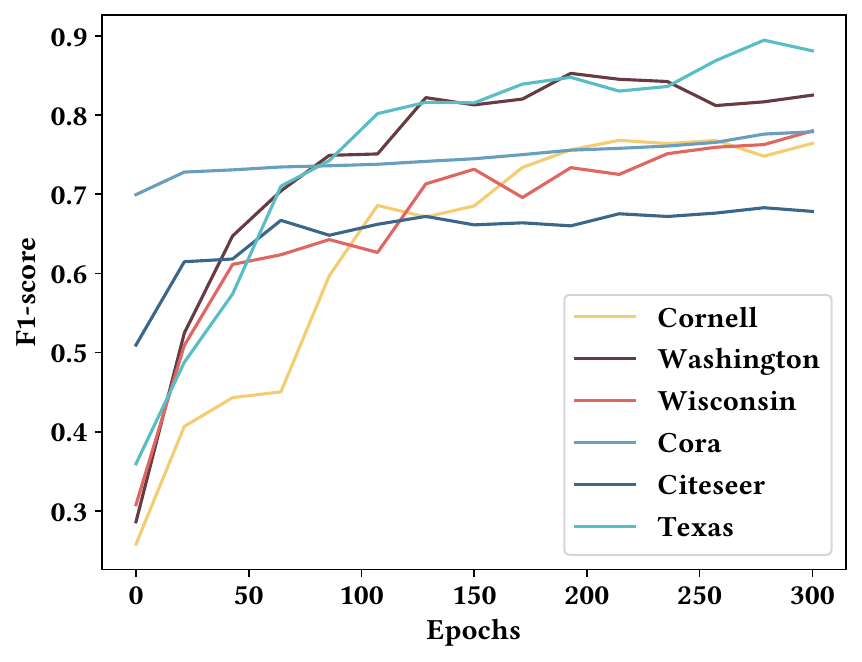}
      }
    \vspace{-3mm}
    \subfigure[{Train loss}]{
        \centering
        \includegraphics[width=0.225\textwidth]{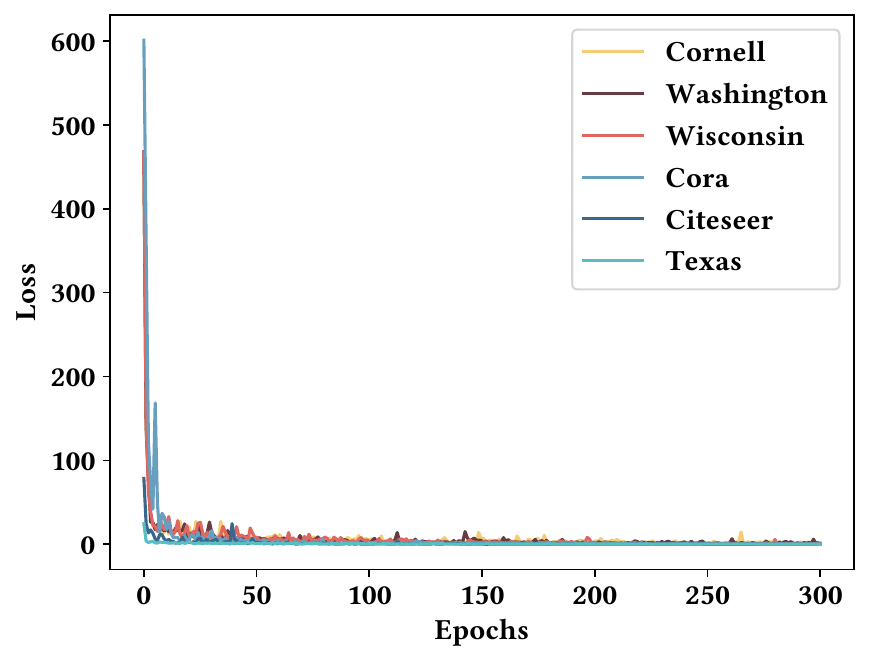} 
    }
    \vspace{-3mm}
    \subfigure[{Varying training set size}]{
        \centering
        \includegraphics[width=0.2257\textwidth]{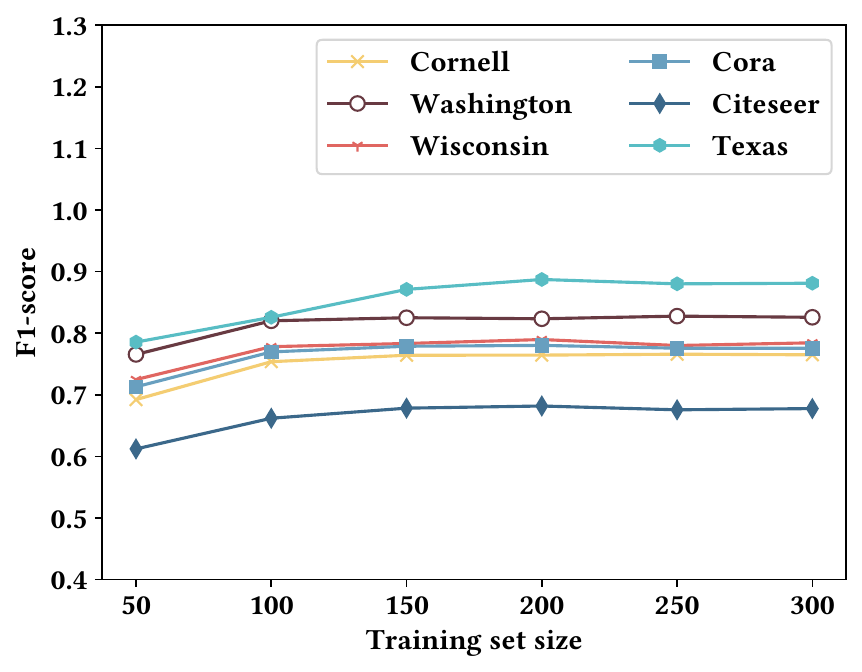} 
    }
    \vspace{-3mm}
    \subfigure[{Varying validation set size}]{
        \centering
        \includegraphics[width=0.225\textwidth]{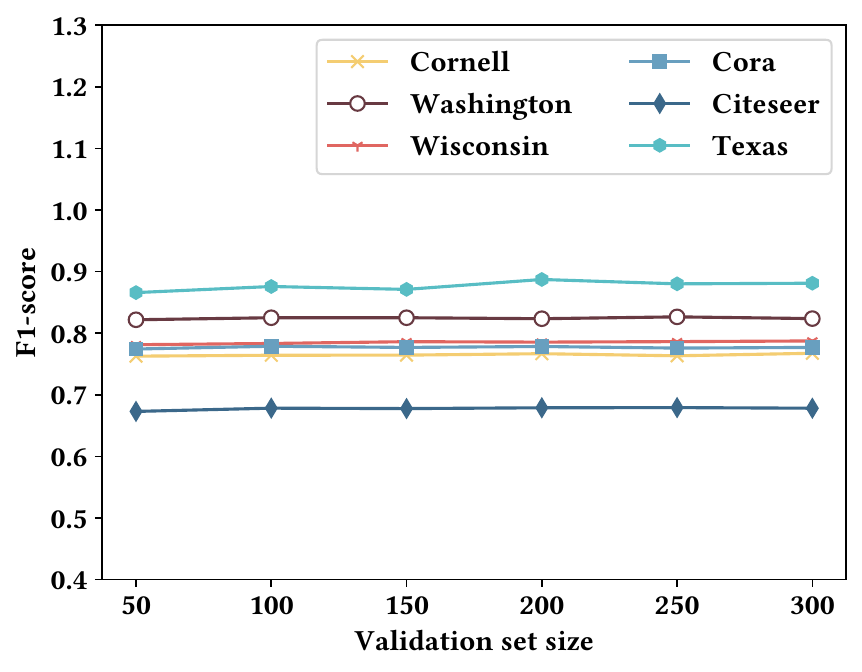} 
    }
    \caption{Experiments on hyper-parameter sensitivity}
    \label{fig:sensitivity}
\vspace{-3mm}
\end{figure}

\begin{figure}
\vspace{-2mm}
    \subfigcapskip=-4pt 
    \subfigure[{AQD-GNN}]{ 
        \centering
        \includegraphics[width=0.20\textwidth]{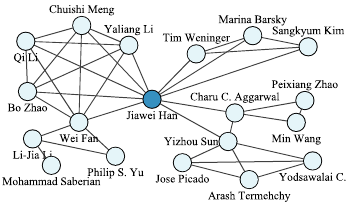}
        \captionsetup{skip=10pt}
      }
    \vspace{-3mm}
    \subfigure[{ALICE}]{
        \centering
        \includegraphics[width=0.22\textwidth]{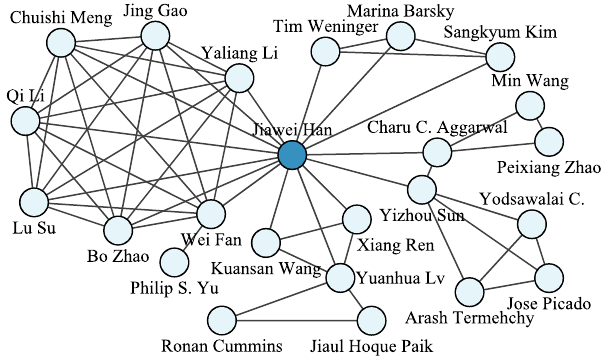} 
    }
     \vspace{-3mm}
    \subfigure[{ALICE} for query \{Jiawei Han, Alan L. Yuille, DB, IR, CV\}]{
        \centering
        \includegraphics[width=0.43\textwidth]{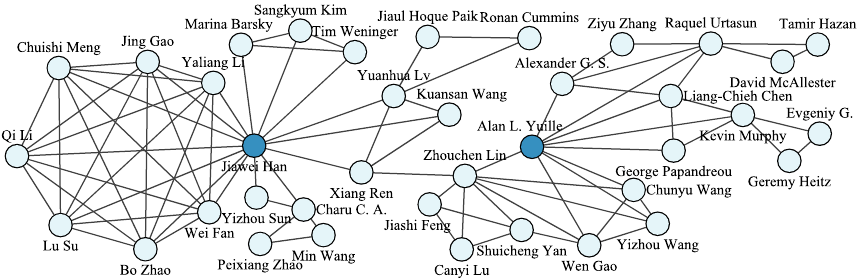} 
    }
    \vspace{-3mm}
    \caption{Case study}
    \label{fig:case_study}
\vspace{-6mm}
\end{figure}

\vspace{-4mm}
\subsection{Case Study}

{\color{black}In this part, we conduct a case study using arXiv~\cite{arxiv_org_submitters_2023} which is a co-author network comprising a diverse collection of scholarly articles from fields such as physics, mathematics, and computer science. We select articles from the field of computer science published between 2011 and 2015, and build a co-author network consisting of 42,065 nodes and 102,165 edges. The network encompasses 40 subdisciplines, e.g., database (DB), and information retrieval (IR). We consider Prof Jiawei Han, a renowned researcher in database and data mining. We use Jiawei Han as a query node and use 2 query attributes, including DB and IR. The return communities are shown in Figure~\ref{fig:case_study}. Specifically, Figure~\ref{fig:case_study} (a) illustrates the community discovered by AQD-GNN and Figure~\ref{fig:case_study} (b) demonstrates the community discovered by \netname. As AQD-GNN is out of memory, we input the candidate subgraph generated by ~\netname~for both methods, with a cap of 400 nodes. We delete unconnected nodes from the result communities. The figure shows that both methods can effectively find the research collaborators of Jiawei Han. Moreover, we find that AQD-GNN misses some promising candidates, e.g., Xiang Ren who is a former student of Jiawei Han and was in close cooperation with Jiawei Han. In contrast, \netname~can discover such an important candidate. Additionally, we further consider Prof Alan L. Yuille, a renowned researcher in computer vision (CV). We use Jiawei Han and Alan L. Yuille as query nodes and use \{DB, IR, CV\} as query attributes. The result is illustrated in Figure~\ref{fig:case_study} (c). We can find that although Jiawei Han and  Alan L. Yuille belong to two different communities under the arXiv network, nodes in the returned community are closely related to the query.}

\vspace{-4mm}
\section{Related Work}
\label{sec:relatedwork}
In this section, we review related studies regarding classical community search and ML/DL for community search.

\noindent \textbf{Classical Community Search.}
Community search aims to find cohesive subgraphs in a graph that contain query nodes and satisfy given constraints. Some well-known classical methods use \ksize-related measurements to model the community such as \ksize-core~\cite{cui2014local, wang2023cohesive, zhang2016engagement}, \ksize-truss~\cite{huang2014querying, wang2020efficient}, \ksize-clique~\cite{wang2022efficient, yuan2017index}, and \ksize-edge connected component (\ksize-ECC)~\cite{chang2015index, hu2016querying}. However, these methods suffer from structure inflexibility (i.e., the real-world community may always dissatisfy with the constraints). Recently, a graph modularity-based method~\cite{DBLP:conf/sigmod/KimLCY22} is proposed, which aims to find a subgraph that has the maximum graph modularity. When targeting attributed graphs, several algorithms have been proposed that consider both structure and keyword cohesiveness, e.g. ACQ~\cite{fang2016effective} and ATC~\cite{huang2017attribute}. Both algorithms adopt a two-stage procedure that first considers the structure constraint and then selects the community with the highest attribute score. However, they fail to capture the correlation between structure and attribute that are closely related.

\noindent \textbf{ML/DL for Community Search.}
With the powerful approximation ability of the neural network, learning-based techniques (like GNNs) have been adopted for community search. The general approach of using GNNs for community search is to model the problem as a binary node classification task, where the goal is to predict the probability of each vertex belonging to a community.  ICS-GNN~\cite{gao2021ics} leverages a GNN model to search community in an iterative manner where the community is modeled as a \ksize-sized subgraph with maximum GNN scores. AQD-GNN~\cite{jiang2022query} is proposed to support the attributed graph that takes both the cohesive structure and homogeneous attributes into account. However, these methods suffer from severe efficiency issues.
\vspace{-4mm}
\section{Conclusion}
\label{sec:conclusion}

In this paper, we explore the problem of attributed community search. 
To improve the accuracy and efficiency, we propose a novel model \netname~that first extracts the candidate subgraph and then predicts the community based on the query and candidate subgraph. 
In the candidate subgraph extraction phase, we design a new modularity named density sketch modularity and adaptively select a reasonable amount of neighbors considering both structure and attribute. 
In the prediction phase, we devise \aggname~to integrate consistency constraints for the prediction of the attributed community. It utilizes a cross-attention encoder to encode the interaction information between the query and the data graph. Structure-attribute consistency and local consistency are utilized to guide the training of the model.
We conduct various experiments over 11 real-world benchmark datasets from multiple aspects. The results demonstrate that \netname~shows a good performance in terms of prediction accuracy, efficiency, robustness, and scalability.


\balance
{
\bibliographystyle{ACM-Reference-Format}
\bibliography{sample}
}

\end{document}